\documentclass[10pt,twocolumn,twoside]{IEEEtran} 
\ifCLASSINFOpdf
\usepackage[cmex10]{amsmath}
\usepackage[all]{xy}
\usepackage{color}
\usepackage{multirow}
\usepackage{algorithm}
\usepackage{cite}
\usepackage{algpseudocode}
\usepackage{subcaption}
\usepackage{bm}
\usepackage{graphicx}
\usepackage{color}
\usepackage{wrapfig}
\usepackage{epsf}
\usepackage{times}
\usepackage{url}
\usepackage{amssymb, amsmath, mathtools, amsthm}
\usepackage{nomencl}
\usepackage{comment}
\usepackage{enumitem}
\usepackage{mathtools}
\usepackage{setspace}
\usepackage{lipsum}

\newcommand\blfootnote[1]{%
 \begingroup
 \renewcommand\thefootnote{}\footnote{#1}%
 \addtocounter{footnote}{-1}%
 \endgroup
}

\makenomenclature
\newtheorem{theorem}{Theorem}
\newtheorem{definition}{Definition}

\newtheorem{lemma}{Lemma}
\newtheorem{assumption}{Assumption}
\definecolor{RED}{rgb}{0.6,0.,0.}
\definecolor{BLUE}{rgb}{0.,0.,0.6}
\definecolor{GREEN}{rgb}{0.,0.6,0.}
\definecolor{MALINA}{rgb}{0.6,0.,0.6}
\definecolor{YELLOW}{rgb}{0.8,0.8,0}
\newcommand{\squeezeup}{\vspace{-1.5 mm}}
\algdef{SE}[DOWHILE]{Do}{doWhile}{\algorithmicdo}[1]{\algorithmicwhile\ #1}
\begin{document}
\title{Learning with End-Users in Distribution Grids:\\Topology and Parameter Estimation}
\author{Sejun Park$^{*}$,~Deepjyoti Deka$^\dag$,~Scott Backhaus$^\ddag$,~Michael Chertkov$^{\mathsection}$\\
$*$KAIST, Daejeon, Korea\\
$\dag$Theoretical Division, Los Alamos National Laboratory, Los Alamos, New Mexico, USA\\
$^\ddag$ Quantum  Electromagnetics  Group, National Institute of Standards and Technology, Boulder, Colorado, USA\\
$^\mathsection$ Dept. of Mathematics, University of Arizona, Tucson, Arizona, USA}

\maketitle
\begin{abstract}
Efficient operation of distribution grids in the smart-grid era is hindered by the limited presence of real-time nodal and line meters. In particular, this prevents the easy estimation of grid topology and associated line parameters that are necessary for control and optimization efforts in the grid. This paper studies the problems of topology and parameter estimation in radial balanced distribution grids where measurements are restricted to only the leaf nodes and all intermediate nodes are unobserved/hidden. To this end, we propose two exact learning algorithms that use balanced voltage and injection measured only at the end-users. The first algorithm requires time-stamped voltage samples, statistics of nodal power injections and permissible line impedances to recover the true topology. The second and improved algorithm requires only time-stamped voltage and complex power samples to recover both the true topology and impedances without any additional input (e.g., number of grid nodes, statistics of injections at hidden nodes, permissible line impedances). We prove the correctness of both learning algorithms for grids where unobserved buses/nodes have a degree greater than three and discuss extensions to regimes where that assumption doesn't hold. Further, we present computational and, more importantly, the sample complexity of our proposed algorithm for joint topology and impedance estimation. We illustrate the performance of the designed algorithms through numerical experiments on the IEEE and custom power distribution models. 
\end{abstract}
\blfootnote{
The authors acknowledge the support from the Department of Energy through the Grid Modernization Lab Consortium and the Center for Nonlinear Studies (CNLS) at Los Alamos National Laboratory for this work.}
\section{Introduction}
\label{sec:intro}
Distribution grids include the low and medium voltage transmission lines that help transfer power from the distribution substation to the final consumers. Structurally, a majority of distribution grids are radial in structure. However, unlike traditional passive distribution grids, modern ones have smart controllable loads, household renewable generators (e.g., solar panels), and battery storage devices (e.g., electric vehicles). The presence of active devices has made distribution grids dynamic, re-configurable, and an important location for smart grid operations like demand response, frequency regulation and inter-household energy settlements/transactions. However, optimal operations under different regimes require real-time state estimation in the grid, in particular of the current radial topology of current operational lines, and their impedances. In addition, real or near real-time estimation of the distribution grid topology and corresponding line impedances is not straightforward due to the limited availability of real-time measurement devices, unlike in high voltage transmission grids. In recent years, Phasor Measurement Unit (PMU) technology and its alternatives (e.g., micro-PMUs \cite{micropmu}, FNETs \cite{FNET}) have become available in distribution grids, but their presence is not ubiquitous \cite{hoffman}. Among others, the presence of underground lines in urban areas makes meter placement, direct estimation, and calibration of parameters challenging. Thus, there is a greater need to develop efficient algorithms that can provably estimate topology and line parameters under sparse meter presence and infrequent calibration of line parameters. More importantly, new loads such as smart air-conditioners or electric vehicles connected to the grid at the end-user level have the ability to measure and communicate nodal voltages and injections. In this work, we consider such scenarios and analyze the topology and parameter estimation problem in grids where only leaf nodes measurements of the grids are available. 

\subsection{Prior Work}
Past research in topology or parameter estimation has proposed different algorithms that differ primarily on the availability of data and type of measurements (nodal or line based). For available line measurements, \cite{ramstanford} uses a cycle basis and maximum likelihood tests to estimate the topology. When nodal voltage based measurements are available at all nodes, graphical model based formulations have been proposed to estimate the operational lines for both radial grids \cite{bolognani2013identification,dekathreephase} and loopy grids \cite{ram_loop,dekairep}. In a similar measurement regime including nodal voltages, \cite{dekatcns} present greedy topology learning schemes based on trends in second moments of voltage magnitudes. Real-data driven and model-free schemes using signature based tests to reconstruct topology and line parameters are presented in \cite{cavraro2015data,arya,sandia1}.

It is worth mentioning that the majority of the prior work relies on the availability of nodal measurements (voltage and/or injection) at all nodes of the grid. In particular, in work involving missing nodes \cite{dekatcns}, injection statistics at all nodes are assumed to be known. However, this might be a strong assumption due to unavailable meters and historical information for missing nodes. In addition, none of the mentioned work provides guaranteed topology and impedance estimation in the presence of missing nodes. In this paper, we provide efficient algorithms for both topology and impedance estimation on all operational lines in the grid in a severely measurement deficient regime where voltage and/or injection measurements at only end-users (leaf nodes in radial grids) are observed while all other nodal quantities are unobserved. 

\vspace{-0.1in}
\subsection{Contribution}
In this paper, we address to estimate topology and impedance on all operational lines only using the leaf node measurements in a balanced radial grid. Unlike prior works, we assume that all other nodal quantities of missing nodes are not available.
To this end, we propose two algorithms for topology and impedance recovery with only leaf node measurements in a radial grid that are provably correct for linearized power flows \cite{89BWa,bolognani2016existence, dekatcns}. 

The first algorithm utilizes time-stamped voltage magnitude samples and complex injection statistics of leaf nodes and identifies operational edges from an over-complete set of permissible edges with known impedances. Operational edges and intermediate missing nodes are identified based on a novel relationship between second order moments of voltage and power injections at the observed leaf nodes. We show that the algorithm has $O(|\mathcal V|^3)$ computational complexity.

We further improve the first algorithm to present our second algorithm that jointly estimates operational edges and their impedances only using time-stamped voltage magnitude and injection samples of leaf nodes. Unlike the first algorithm, the second algorithm does not require any knowledge of the missing nodes or the permissible lines. 
The second algorithm first recovers the impedance distance between all observed leaf nodes and iteratively identifies each operational edge along with its impedance. The second algorithm also has $O(|\mathcal V|^3)$ computational complexity. In addition, we prove that it has $O(|\mathcal V|\log|\mathcal V|)$ sample complexity for the correct recovery of the topology.
Simulations results on IEEE test cases with ac power flow models demonstrate the practical use of our algorithms. To the best of our knowledge, this is the first work which provides guaranteed topology and impedance reconstruction in balanced distribution grids, only using the leaf nodes measurements. 
We present a summary of both algorithms in Table \ref{table:algcomp}.
Parts of the work have been presented in IEEE SmartGridComm 2016 \cite{dekasmartgridcomm} and PSCC 2018 \cite{sejunpscc}. This journal version includes new theoretical results on sample complexities that prove the correctness for the performance of the algorithms at finite samples. Further, we include a detailed discussion for the extension of the developed algorithms and additional simulation results on ac power flow based samples.

\begin{table*}
\centering
\caption{Summary of topology learning algorithms with missing nodes}
\label{table:algcomp}
 {\footnotesize
 \begin{tabular}{ | c | c | c | c | c |}
 \hline
 Algorithm & Output & Observations (\textbf{available only at leaf nodes}) & Prior Information & Assumptions \\ 
 \hline
 Algorithm \ref{alg:learningdeep} & Topology & Time-stamped voltage magnitude samples & Line impedances & Uncorrelated power injections\\ 
 & & Complex power injection statistics&of all permissible lines & Missing nodes have degree $\ge 3$ \\ \hline
 Algorithm \ref{alg:main}& Topology & Time-stamped voltage magnitude samples &None & Uncorrelated power injections\\ 
 & Impedances &Time-stamped complex power injection samples& & Missing nodes have degree $\ge 3$ \\ 
 \hline
 \end{tabular}
 }
 \vspace{-0.1in}
\end{table*}
The rest of the paper is organized as follows. Section \ref{sec:preliminary} introduces nomenclature and power flow relations in the distribution grid. The first algorithm and the second algorithm are described in Section \ref{sec:firstalg} and Section \ref{sec:main} respectively. Numerical experiments are presented in Section \ref{sec:experiments}. Finally, Section \ref{sec:conclusions} is reserved for conclusions and discussions.

\section{Distribution Grid Topology and Power Flows}\label{sec:preliminary}
\begin{figure}[ht]
 \centering
\includegraphics[width=0.27\textwidth]{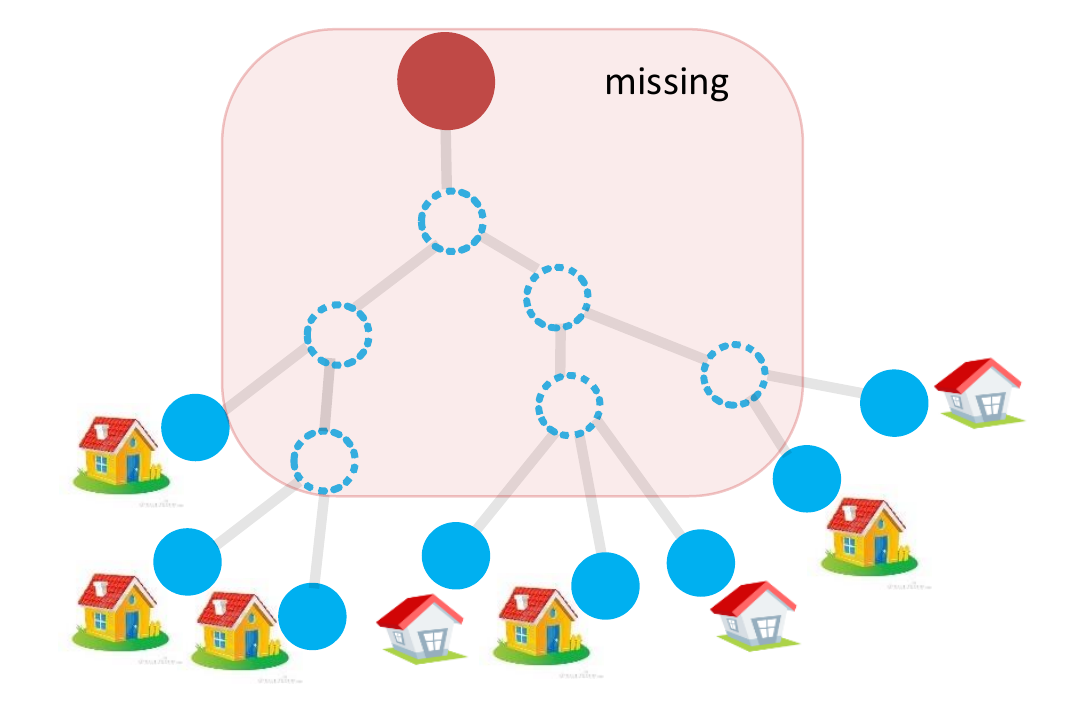}
\caption{Illustration of a radial distribution grid. The red node denotes the substation bus. Blue nodes denote observed end-users. Dotted nodes denote unobserved and unknown intermediate nodes. Grey lines denote unknown operational edges.
}
 \label{fig:radialgrid}
\squeezeup
\end{figure}
\noindent{\bf Radial Topology:}
We consider radial distribution grids in this paper. Mathematically, we define a grid by a graph $\mathcal G=(\mathcal V,\mathcal E)$, where the set of buses/nodes is denoted by $\mathcal V$ and the set of undirected operational lines/edges is denoted by $\mathcal E$.
We denote the set of all lines/edges including non-operational ones by $\bar{\mathcal E}$.
The operational grid is `radial' if $\mathcal G$ is tree-structured. Fig. \ref{fig:radialgrid} shows an illustration of a distribution grid. We use lower-case alphabets $a,b,c,\dots$ to represent buses/nodes and a pair $(ab)$ to denote a line/edge between nodes $a$ and $b$. We denote $t\in\mathcal V$ as a root node (reference/substation bus). The case with multiple substations is discussed in later sections as an extension of the algorithms in the paper.
We denote $\mathcal P_{ab}$ as the unique path from a node $a$ to a node $b$ in a radial grid $\mathcal G$. For a node $a$, all nodes whose path to the root contains $a$ is called descendants of $a$ and denoted by a set $\mathcal{D}_a$. If $(ab)$ is an edge and $b \in \mathcal{D}_a$, then a node $b$ is called `child' of a `parent' node $a$. Nodes that are children of a same parent are called `siblings'. $\mathcal L\subset\mathcal V$ denotes the set of leaf nodes that are observed in our learning algorithms. The remaining intermediate/missing nodes in the grid are assumed to be unobserved. Next, we discuss the power flow models used in this paper for our algorithm designs. 

\noindent{\bf Power Flow Models:}
In a radial grid $\mathcal G=(\mathcal V,\mathcal E)$, we consider the balanced power flow operation satisfying the following Kirchhoff's law which expresses complex power injections at each node in terms of the node-voltages and line impedances:
\begin{equation}\label{eq:powerflow}
p_a+iq_a=\sum_{b:(ab)\in\mathcal E}\frac{v_a^2-v_av_b\exp(i\theta_a-i\theta_b)}{z^*_{ab}}.
\end{equation}
Here, $z_{ab},v_a,\theta_a,p_a,q_a$ denote impedance of $(ab)\in\mathcal E$, balanced voltage magnitude, voltage phase, active and reactive power at $a\in\mathcal V$, respectively. 
Since Eq.~(\ref{eq:powerflow}) is non-convex, we consider a linearized approximation that neglects second order terms in Eq.~(\ref{eq:powerflow}) termed \textbf{Linear Coupled Power Flow (LC-PF)} model \cite{bolognani2016existence,dekatcns}:
\begin{equation}\label{eq:lcpf}
\begin{split}
&p_a=\sum_{b:(ab)\in\mathcal E}\big[\beta_{ab}(\theta_a-\theta_b)+g_{ab}(v_a-v_b)\big]\\
&q_a=\sum_{b:(ab)\in\mathcal E}\big[\beta_{ab}(v_a-v_b)-g_{ab}(\theta_a-\theta_b)\big]
\end{split}
\end{equation}
where $g_{ab}=r_{ab}/(x_{ab}^2+r_{ab}^2)$, $\beta_{ab}=x_{ab}/(x_{ab}^2+r_{ab}^2)$ and $r_{ab},x_{ab}$ are resistance, reactance of $(ab)$, respectively, i.e., $z_{ab}=r_{ab}+ix_{ab}$. Following the standard notation, we consider the substation/root node as a reference bus and measure voltage magnitude and phase at each non-substation bus with respect to it. Further, due to the lossless nature of the linearized power flow model, the injection at the reference bus is the negative of the sum of injections at all other nodes. One can thus ignore the reference bus from the power flow analysis and consider a reduced model comprising of power flow equations at the non-reference buses in the grid. Further, by considering only deviations from the respective steady state values, we model $p,q,v,\theta$ as zero mean random variables. The LC-PF model is equivalent to a first order approximation of voltage magnitudes in the LinDistFlow equations introduced in \cite{89BWa} for distribution grids. The LC-PF model Eq.~(\ref{eq:lcpf}) can also be stated in the following matrix form \cite{dekatcns}
\begin{align}\label{eq:lcpf2}
&v=H^{-1}_{1/r}p+H^{-1}_{1/x}q\qquad \theta=H^{-1}_{1/x}p-H^{-1}_{1/r}q
\end{align}
where $v,\theta,p,q$ are respectively vectors of voltage magnitude, voltage phase, active and reactive power at the non-substation buses of the grid. 
$H_{1/r},H_{1/x}$ represent the reduced weight Laplacian matrices for $\mathcal G\setminus\{t\}$ where $1/r_{ab},1/x_{ab}$ are used as edge-weights of $(ab)$ respectively.\footnote{$\mathcal G\setminus\{t\}$ denotes a subgraph of $\mathcal G$ induced by $\mathcal V\setminus\{t\}$.} We mention a structural property of $H^{-1}_{1/r},H^{-1}_{1/x}$ that arises due to the radial topology.
\begin{lemma}[\cite{68Resh,dekatcns}]\label{lemma1}
Let $H_{1/r}$ be the reduced weighted Laplacian matrix of a grid $\mathcal G$. Then, its inverse satisfies
\begin{align}
 H_{1/r}^{-1}(a,b)&= \sum_{(cd) \in {\cal P}_{at}\bigcap {\cal P}_{bt}} r_{cd}.\label{Hrxinv}
\end{align}
\end{lemma}
Thus, the $(a,b)^{th}$ entry in $H^{-1}_{1/r}$ is equal to the sum of line resistances on edges common to paths from node $a$ and $b$ to the root. As ${\cal P}_{at} \subset {\cal P}_{bt}$ for a parent-child pair $a,b$, Eq.~(\ref{Hrxinv}) gives the following for a parent node $a$ and a child node $b$ for all $c$.
\begin{align}
{\huge H}_{1/r}^{-1}(a,c)-{\huge H}_{1/r}^{-1}(b,c) =\begin{cases}r_{ab} & \text{if $c$ is a descendant of $b$}\\
0 & \text{otherwise} \end{cases}.
\end{align}
\vspace{-0.3in}
\subsection{Assumptions on Distribution Grid Topology}
We now present two assumptions on the distribution grid topology and statistics of power injections, required for the correctness of our topology/impedance learning algorithms.
\begin{assumption}\label{asm:deg3}
All missing nodes have degrees at least 3.
\end{assumption}
Assumption \ref{asm:deg3} implies that each missing node has at least two children and that all leaf nodes are observed. In its absence, the system is thus under-determined and multiple configurations satisfy the available measurements (see \cite{dekasmartgridcomm}).
In later sections, we discuss topology learning without Assumption \ref{asm:deg3}. Assumption \ref{asm:deg3} is akin to assumptions for recovery in graphical models \cite{pearl}. In addition, we assume that the complex power injections at different nodes are uncorrelated.
\begin{assumption}\label{asm:indep}
Injections at all non-substations nodes are modeled as $PQ$ loads with $\mathbb{E}[p_ap_b]=\mathbb{E}[q_aq_b]=\mathbb{E}[p_aq_b]=0~~\forall a\ne b$.
\end{assumption}
As considered in prior studies \cite{bolognani2013identification, dekatcns}, Assumption \ref{asm:indep} is well-justified over sufficiently short time intervals while considering deviations of injections at end-users. For intermediate nodes that are involved in the separation of power into downstream lines, leakage, or device losses contribute to the net power injection and are independent of other nodes. In particular, note that Assumption \ref{asm:indep} does not restrict the class of distributions that can be used to model individual node’s power injection and applies for both positive and negative nodal injections. We discuss techniques to extend our work to cases with correlated user injection profiles and multi-phase systems in future works.

\vspace{-0.1in}
\section{Topology Learning Algorithm with Voltage Samples}\label{sec:firstalg}
In this section, we discuss properties of voltages and injections at leaf nodes, and utilize them to design the first topology learning algorithm, Algorithm \ref{alg:learningdeep}, introduced in \cite{dekasmartgridcomm}. Algorithm \ref{alg:learningdeep} utilizes voltage samples and injection statistics at all leaf nodes and identifies all operational edges from an over-complete set of permissible edges $\bar{\mathcal E}$ with known impedances. Using the LC-PF Eq.~(\ref{eq:lcpf2}), one can write the second moments of nodal voltages with that of nodal injections as:
\begin{align}\label{eq:expand}
\mathbb{E}[vv^T]&=H^{-1}_{1/r}\mathbb{E}[pp^T]H^{-1}_{1/r}+H^{-1}_{1/x}\mathbb{E}[qq^T]H^{-1}_{1/x} \nonumber\\
&~+H^{-1}_{1/r}\mathbb{E}[pq^T]H^{-1}_{1/x}+H^{-1}_{1/x}\mathbb{E}[qp^T]H^{-1}_{1/r}
\end{align}
Note that $\mathbb{E}[pp^T],\mathbb{E}[qq^T],\mathbb{E}[qp^T],\mathbb{E}[pq^T]$ are diagonal matrices from Assumption \ref{asm:indep}. For the notational convenience, we first define the variance $\phi_{ab}$ of the difference of voltage measurements at $a,b\in\mathcal V$ as follow:
\begin{equation}\label{eq:phi}
\phi_{ab}=\mathbb{E}[(v_a-v_b)^2]=\mathbb{E}[v_a^2]+\mathbb{E}[v_b^2]-2\mathbb{E}[v_av_b].
\end{equation}
Note that $\phi_{ab}$ can be estimated for all leaf pairs in $\mathcal L$ using the observed voltage samples. Under LC-PF model, the following result holds:
\begin{theorem}\label{thm:phiab}
Let $a,b\in\mathcal L$ have a common parent $k_1$. Then
\begin{align}
\phi_{ab}&=r_{ak_1}^2\mathbb{E}[p_a^2]+x_{ak_1}^2\mathbb{E}[q_a^2]+2r_{ak_1}x_{ak_1}\mathbb{E}[p_aq_a]\notag\\
&\,+r_{k_1b}^2\mathbb{E}[p_b^2]+x_{k_1b}^2\mathbb{E}[q_b^2]+2r_{k_1b}x_{k_1b}\mathbb{E}[p_bq_b].\label{eq:phiab}
\end{align}
\end{theorem}
The derivation follows by expanding $\phi_{ab}$ using Eq.~\eqref{eq:expand} and using Lemma \ref{lemma1}. Note that aside from pathological cases, Theorem \ref{thm:phiab} is satisfied only by the true parent $k_1$ of nodes $a,b$. Thus, the equality can be used to identify the true parent of sibling leaves (see Fig. \ref{fig:alg1}, \ref{fig:alg2} for an example). The next result involves $\phi$ values at three leaf nodes in $\mathcal L$. 

\begin{theorem}[Theorem 2 of \cite{dekasmartgridcomm}]\label{thm:phiacbc}
Let $a,b\in\mathcal L$ have a common parent $k_1$. Consider $c\in\mathcal L$ such that $c,k_1\in\mathcal D_{k_2}$ and $\mathcal P_{k_1t}\cap\mathcal P_{ct}=\mathcal P_{k_2t}$ for some intermediate node $k_2$ (see Fig. \ref{fig:alg6} for an example). Then
\begin{align}
\phi_{ac}-&\phi_{bc}=\mathbb{E}[p_a^2]((r_a^{k_2})^2-(r_{k_1}^{k_2})^2)+\mathbb{E}[q_a^2]((x_a^{k_2})^2-(x_{k_1}^{k_2})^2)\notag\\
&+2\mathbb{E}[p_aq_a](r_a^{k_2}x_a^{k_2}-r_{k_1}^{k_2}x_{k_1}^{k_2})-\mathbb{E}[p_b^2]((r_b^{k_2})^2-(r_{k_1}^{k_2})^2)\notag\\
&-\mathbb{E}[q_b^2]((x_b^{k_2})^2-(x_{k_1}^{k_2})^2)-2\mathbb{E}[p_bq_b](r_b^{k_2}x_b^{k_2}-r_{k_1}^{k_2}x_{k_1}^{k_2})\label{eq:phiacbc}
\end{align}
where $r_d^e=\sum_{(fg)\in\mathcal P_{de}}r_{fg},x_d^{e}=\sum_{(fg)\in\mathcal P_{de}}x_{fg}$.
\end{theorem}
The proof of Theorem \ref{thm:phiacbc} uses algebraic expansions of the expression for $\phi_{ac}-\phi_{bc}$ and application of Lemma \ref{lemma1}. We refer the reader to check \cite{dekasmartgridcomm} for the details. Crucially, Theorem \ref{thm:phiacbc} enables us to identify edges between missing intermediate nodes. For example, consider the case in Fig. \ref{fig:alg3}. If edges from node $k_1$ to leaves $a,b$ have already been discovered, then one can assert the existence of edge $(k_1k_2)$ to $k_1$'s parent $k_2$ by checking if Eq.~\eqref{eq:phiacbc} holds. We use this result in our algorithm to learn edges iteratively from parents of leaves to the root. However, it needs to be mentioned that the right side of Eq.~\eqref{eq:phiacbc} does not depend on the path to a node $c$. Thus, $c$ can only be identified as a descendant of $k_2$. Its true location cannot be identified, in particular, if a leaf node $c$ does not have another leaf node as a sibling. Locations of such leaf nodes are determined once the rest of the grid is recovered. We arrange the identified intermediate edges in reverse order and check for Eq.~(\ref{eq:phiacbc}) to arrive at the true parent of an unidentified leaf node $c$. The post-order node traversal \cite{Cormen2001} is necessary to ensure that the true parent of $c$ is checked before other intermediate nodes on $\mathcal{P}_{ct}$. The steps are outlined in Algorithm \ref{alg:learningdeep}.

Now, we briefly explain Algorithm \ref{alg:learningdeep}, which takes voltage samples and injection statistics at leaf nodes and outputs the set of operational edges $\mathcal{E}$ from the set of input permissible edges $\bar{\mathcal{E}}$ with known impedances. In Steps 5-11, Algorithm \ref{alg:learningdeep} identifies sibling relationships of leaf nodes and find their parent using Theorem \ref{thm:phiab}. In Steps 12-24, the algorithm identifies edges between missing nodes using Theorem \ref{thm:phiacbc}, as explained in the example in the last paragraph. Similarly, in Steps 26-33, the algorithm finds parents of leaf nodes without sibling leaf nodes. Fig.~\ref{fig:alg} illustrates Algorithm \ref{alg:learningdeep} step by step.
\begin{figure*}
\centering
 \begin{subfigure}{0.12\textwidth}
 \includegraphics[width=\textwidth]{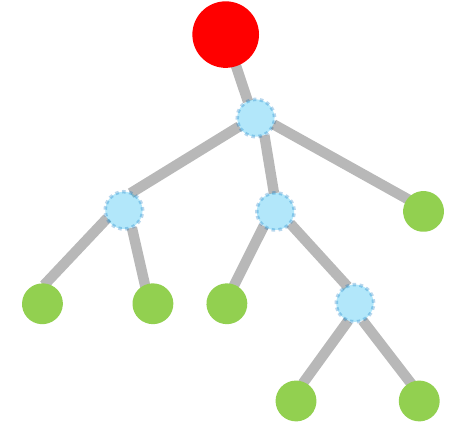}
 \caption{}
 \label{fig:alg1}
 \end{subfigure}
 ~
 \begin{subfigure}{0.12\textwidth}
 \includegraphics[width=\textwidth]{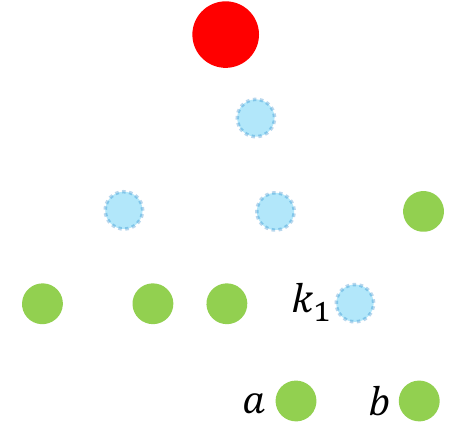}
 \caption{}
 \label{fig:alg2}
 \end{subfigure}
 ~
 \begin{subfigure}{0.12\textwidth}
 \includegraphics[width=\textwidth]{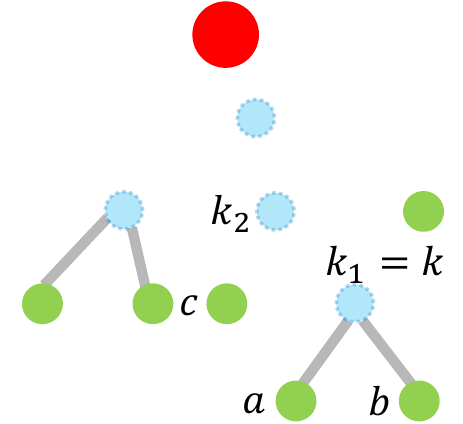}
 \caption{}
 \label{fig:alg3}
 \end{subfigure}
 ~
 \begin{subfigure}{0.12\textwidth}
 \includegraphics[width=\textwidth]{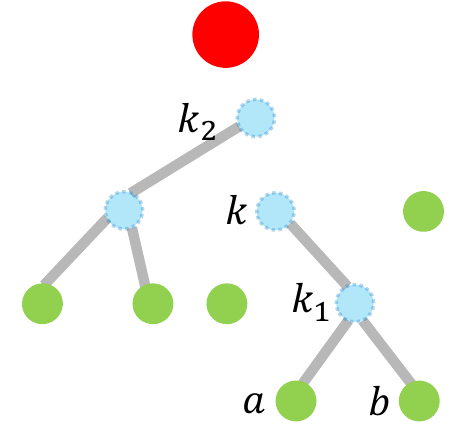}
 \caption{}
 \label{fig:alg4}
 \end{subfigure}
 ~
 \begin{subfigure}{0.12\textwidth}
 \includegraphics[width=\textwidth]{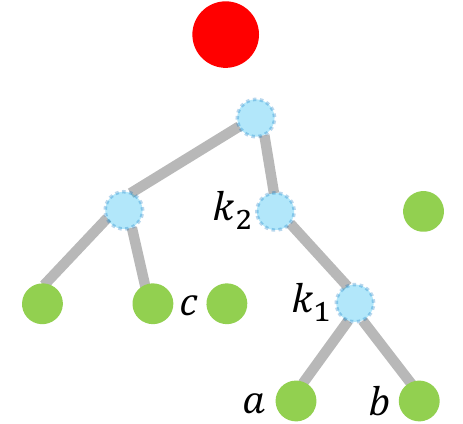}
 \caption{}
 \label{fig:alg5}
 \end{subfigure}
 ~
 \begin{subfigure}{0.12\textwidth}
 \includegraphics[width=\textwidth]{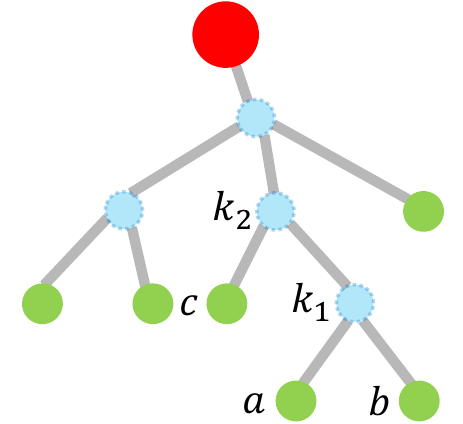}
 \caption{}
 \label{fig:alg6}
 \end{subfigure}
\caption{Illustration of Algorithm \ref{alg:learningdeep} (a) an original topology with missing blue nodes and observed green nodes (b \& c) finding parent $k_1$ of sibling leaf nodes $a,c$ (Steps $5-11$) (d \& e) iterative recovery of missing parent and grandparent of $k_1$ (Steps $12-24$) (f) finding parents of leaf nodes without sibling leaves (Steps $25-36$) to get recovered topology
}\label{fig:alg}
\vspace{-0.1in}
\end{figure*}

\begin{algorithm}\caption{Topology Learning Algorithm with Voltage Magnitude Samples}\label{alg:learningdeep}
\begin{algorithmic}[1]
\State {\bf Input:} 
$\mathcal L,\mathcal M=\mathcal V\setminus\mathcal L$, $\{\mathbb{E}[p_a^2],\mathbb{E}[p_aq_a],\mathbb{E}[q_a^2]:a\in\mathcal L\}$, $v^1,\dots,v^m$, 
$\{r_{ab},x_{ab}:(ab)\in\bar{\mathcal E}\}$
\State {\bf Output:} $(\mathcal V,\mathcal E)$
\State {\bf Initialization:} $par_a\leftarrow\emptyset$, $des_a\leftarrow\emptyset$ for all $a\in\mathcal V$, $\mathcal E\leftarrow\emptyset$ 
\State Compute $\phi_{ab}=\mathbb{E}[(v_a-v_b)^2]$ for all $a,b\in\mathcal L$
\For{$a\in\mathcal L$}
\If{$par_a=\emptyset,\exists b\in\mathcal L,\exists k_1\in\mathcal M$ s.t. $a,b,k_1$ satisfy Eq.~(\ref{eq:phiab}) with tolerance $\tau_1$}
\State $\mathcal E\leftarrow\mathcal E\cup \{(ak_1),(bk_1)\}$
\State $par_a\leftarrow\{k_1\}$, $par_b\leftarrow\{k_1\}$, $des_{k_1}\leftarrow\{a,b\}$
\EndIf
\EndFor
\State $\mathcal L\leftarrow\{a:a\in\mathcal L,par_a=\emptyset\}$
\Do
\State $\mathcal M_1\leftarrow\{k:k\in\mathcal M,par_{k}=\emptyset,des_{k}\ne\emptyset\}$
\State $\mathcal M_2\leftarrow\{k:k\in\mathcal M,des_{k}=\emptyset\}$
\For{$k\in\mathcal M_1$ with $a,b\in des_k$}
\State $k_1\leftarrow par_a$
\If{$\exists k_2 \in \mathcal M-\mathcal M_2$, $\phi_{ac}-\phi_{bc}$ satisfies Eq.~(\ref{eq:phiacbc}) with tolerance $\tau_2$ for some $c\in des_{k_2}$}\footnotemark
\State $\mathcal E\leftarrow\mathcal E\cup\{(kk_2)\},par_k\leftarrow\{k_2\}$
\ElsIf{$\exists k_2\in\mathcal M_2$ s.t. $\phi_{ac}-\phi_{bc}$ satisfies Eq.~(\ref{eq:phiacbc}) with tolerance $\tau_2$ for some $c\in\mathcal L$}
\State $\mathcal E\leftarrow\mathcal E\cup\{(kk_2)\},par_k\leftarrow\{k_2\}$
\State $des_{k_2}\leftarrow des_k$
\EndIf
\EndFor
\doWhile{$|\{k_1:k_1\in\mathcal M_1,par_{k_1}\ne\emptyset\}|>0$}
\State Form a post-order traversal node set $\mathcal W$ using $par_a$ for all $a\in\mathcal M$ such that $des_a\ne\emptyset$
\For{$c\in\mathcal L$}
\For{$j=1$ to $|\mathcal W|$}
\State $k_2\leftarrow\mathcal W(j)$ with $a,b\in des_{k_2},k_1\leftarrow par_a$
\If{$\phi_{ac}-\phi_{bc}$ satisfies Eq.~(\ref{eq:phiacbc}) with tolerance $\tau_2$}
\State $\mathcal E\leftarrow\mathcal E\cup\{(ck_2)\},j\leftarrow|\mathcal W|,\mathcal W\leftarrow\mathcal W\setminus\{k_2\}$
\EndIf
\EndFor
\EndFor
\If{$|\mathcal M_1|=1$}
\State Join $k\in\mathcal M_1$ to root
\EndIf
\end{algorithmic}
\end{algorithm}
\footnotetext{Eq.~(\ref{eq:phiacbc}) is checked under the assumption that $(kk_2)$ exists.}

\noindent\textbf{Computational Complexity of Algorithm \ref{alg:learningdeep}:}
Algorithm \ref{alg:learningdeep} has three major parts, Steps 5-11, Steps 12-24 and Steps 26-33 where the rest part has complexity $O(|\mathcal V|^2)$ which arises from Step 4. Steps 5-11 iterate over a set $(a,b,k_1)\in\mathcal L\times\mathcal L\times\mathcal M$ where each iteration takes $O(1)$ computations.
Therefore, the complexity for steps 5-11 is $O(|\mathcal V|^3|)$. Steps 12-24 and steps 26-33 iterate over sets $(c,k_1,k_2)\in\mathcal L\times\mathcal M_1\times\mathcal M$ and $(c,k_2)\in\mathcal L\times\mathcal M$ respectively, where each iteration takes $O(1)$ computations. 
Therefore, the complexity for Steps 12-24 and 29-24 are $O(|\mathcal V|^3|)$ and $O(|\mathcal V|^2|)$, respectively. Hence, the overall computational complexity of Algorithm \ref{alg:learningdeep} is $O(|\mathcal V|^3)$.

\noindent\textbf{Modification for Finite Samples}: Note that in reality, due to finite samples, the equality relations Eqs.~(\ref{eq:phiab}, \ref{eq:phiacbc}) will not hold with equality. In that setting, we compute the relative difference between the left and right sides for either relation Eq.~(\ref{eq:phiab}) or Eq.~(\ref{eq:phiacbc}). We consider the relations to be satisfied in Algorithm \ref{alg:learningdeep} if the relative differences are respectively less than user defined tolerances $\tau_1, \tau_2$.

The main bottleneck of Algorithm \ref{alg:learningdeep} is that it requires knowledge about permissible edges and impedances which might not be available in real distribution grids. To overcome this, we propose a new algorithm in the next section, which only requires leaf node measurements for recovering the true topology and line impedances on operational edges. 

\section{Topology and Impedance Learning Algorithm with Voltage Magnitude and Power Samples}\label{sec:main}
Our algorithm, termed Algorithm \ref{alg:main}, utilizes time-stamped observations of voltage magnitudes and complex power injections at the end-nodes as the input. Our algorithm mainly utilizes the notion of additive `distance' defined as a distance over the graph, which satisfies the weighted metric property, $d(a,b)=
\sum_{(cd)\in\mathcal{P}_{ab}}d(c,d)$. We first estimate this distance between all leaf node pairs, and then utilize the recursive grouping algorithm \cite{choi2011learning} to learn the operational topology of the grid. Under Assumption \ref{asm:indep} and using the LC-PF Eq.~(\ref{eq:lcpf2}) for observed nodes $a,b$, we derive the following identity
\begin{equation}\label{eq:dist3}
\begin{split}
&\mathbb{E}[v_ap_b]=H^{-1}_{1/r}(a,b)\mathbb{E}[p_b^2]+H^{-1}_{1/x}(a,b)\mathbb{E}[p_bq_b]\\
&\mathbb{E}[v_aq_b]=H^{-1}_{1/r}(a,b)\mathbb{E}[p_bq_b]+H^{-1}_{1/x}(a,b)\mathbb{E}[q_b^2]
\end{split}
\end{equation}
where $\mathbb{E}[v_ap_b],\mathbb{E}[v_aq_b],\mathbb{E}[p_b^2],\mathbb{E}[p_bq_b],\mathbb{E}[q_b^2]$ can be computed from measurements at observed nodes $a$ and $b$. Using Eq.~(\ref{eq:dist3}), one can estimate the value of $H^{-1}_{1/r}(a,b)$ and $H^{-1}_{1/x}(a,b)$ for any observed $a,b\in\mathcal V$ unless $\mathbb{E}[p_b^2]\mathbb{E}[q_b^2]=(\mathbb{E}[p_bq_b])^2$. To avoid such pathological situations, we make the following assumption.
\begin{assumption}\label{asm:bound}
There exists a constant $D>0$ such that for all node $a\in\mathcal V$, $\big|\mathbb{E}[p_a^2]\mathbb{E}[q_a^2]-(\mathbb{E}[p_aq_a])^2\big|\ge D.$
\end{assumption}
Once $H^{-1}_{1/r}(a,b)$ is estimated, one can derive the resistance distance (effective resistance) between observed nodes $a,b$ as
\begin{equation}\label{eq:dist4}
d_r(a,b)=\smashoperator[lr]{\sum_{(cd)\in\mathcal
P_{ab}}}r_{cd}=H^{-1}_{1/r}(a,a)+H^{-1}_{1/r}(b,b)-2H^{-1}_{1/r}(a,b)
\end{equation}
Note that for radial grids, the effective resistance is an additive distance metric between nodes $a$ and $b$ in the grid. Similarly, one can also estimate the additive reactance distance $d_x(a,b)$. Following estimation of $d_r(a,b)$ for all pairs of observed nodes, we utilize the recursive grouping algorithm (RG) \cite{choi2011learning}, which leads to consistent topology and impedance estimation of the power grid $\mathcal{G}$ under Assumption \ref{asm:deg3}.

\begin{figure*}
\centering
 \begin{subfigure}{0.12\textwidth}
 \includegraphics[width=\textwidth]{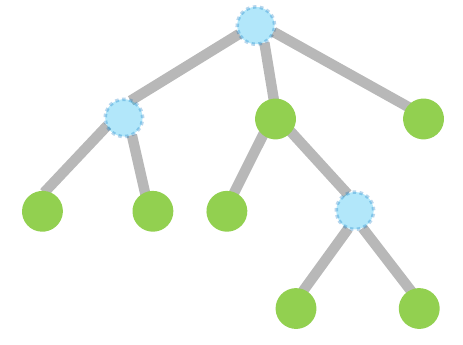}
 \caption{}
 \label{fig:rg1}
 \end{subfigure}
 ~
 \begin{subfigure}{0.12\textwidth}
 \includegraphics[width=\textwidth]{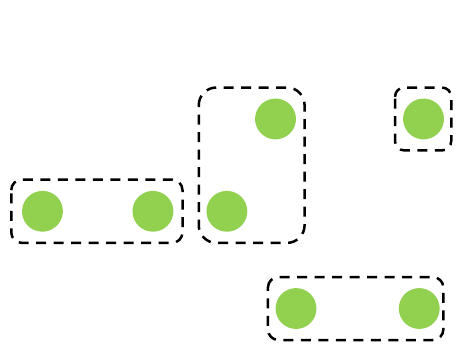}
 \caption{}
 \label{fig:rg2}
 \end{subfigure}
 ~
 \begin{subfigure}{0.12\textwidth}
 \includegraphics[width=\textwidth]{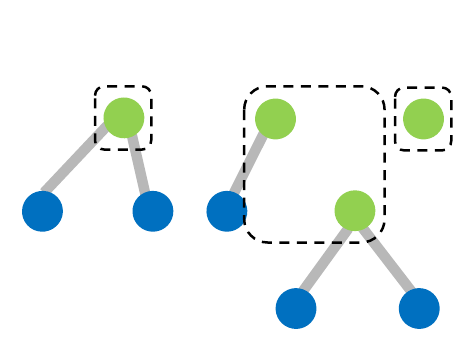}
 \caption{}
 \label{fig:rg3}
 \end{subfigure}
 ~
 \begin{subfigure}{0.12\textwidth}
 \includegraphics[width=\textwidth]{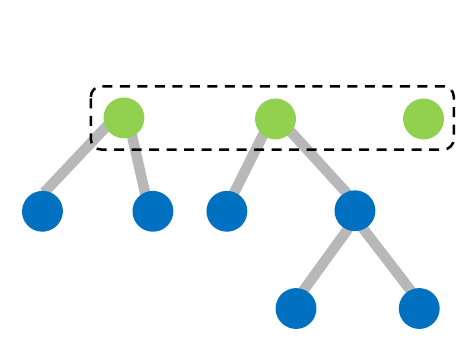}
 \caption{}
 \label{fig:rg4}
 \end{subfigure}
 ~
 \begin{subfigure}{0.12\textwidth}
 \includegraphics[width=\textwidth]{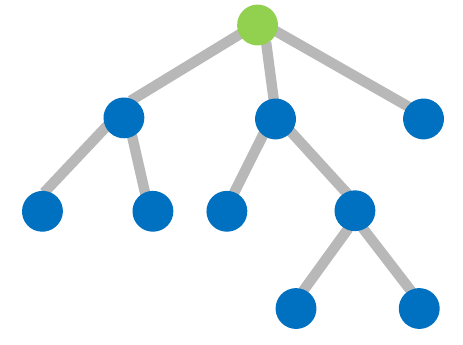}
 \caption{}
 \label{fig:rg5}
 \end{subfigure}
 ~
 \begin{subfigure}{0.12\textwidth}
 \includegraphics[width=\textwidth]{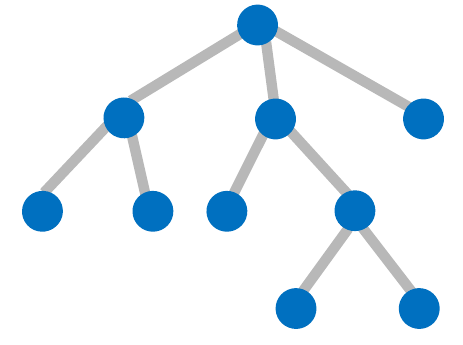}
 \caption{}
 \label{fig:rg6}
 \end{subfigure}
 \caption{Illustration of Algorithm \ref{alg:rg} (a) an original topology with blue missing nodes and green observed nodes ($\mathcal O$). Note that Algorithm \ref{alg:rg} works both with or without internal observed nodes. (b) a partition $\Pi$ (dashed boxes) of $\mathcal O$ generated by a node, its siblings and parent in the first iteration of RG (c) addition of edges and missing nodes and updated $\mathcal O$ after the first iteration of RG and partitions in the second iteration (d) updated $\mathcal O$ after the second iteration and partition of $\mathcal O$ in the third iteration (e) result after the third iteration of RG (f) the recovered topology
 }
 \label{fig:rg}
 \vspace{-0.1in}
\end{figure*}
\subsection{Recursive Grouping Algorithm}
Here, we introduce the recursive grouping (RG) algorithm that recovers the true radial topology given any additive distance $d(\cdot,\cdot)$ between all leaf nodes.
Let us first assume that the exact values of $d(\cdot,\cdot)$ are known for all pairs of observed nodes. Under this assumption, RG utilizes the following lemma \cite{choi2011learning} for the topology and impedance recovery.
We will extend this to the noisy $d(\cdot,\cdot)$ case in Section \ref{sec:finite}.
We note that `parent', `child' in algorithms and lemmas in this section is not related to the substation node as defined in Section \ref{sec:preliminary}.
\begin{lemma}[Lemma 4 of \cite{choi2011learning}]\label{lem:dist}
For $\Phi_{abc}:=d(a,c)-d(b,c)$, the following relation holds:
\begin{itemize}
 \item[a)] $\Phi_{abc}=d(a,b)$ for all $c\in\mathcal V\setminus\{a,b\}$ if and only if $a$ is a leaf node and $b$ is its parent.
 \item[b)] $-d(a,b)\le \Phi_{abc}=\Phi_{abc^\prime}\le d(a,b)$ for all $c,c^\prime\in\mathcal V\setminus\{a,b\}$ if and only if $a,b$ are leaf nodes with common parent, i.e., they belong to the same group of siblings.
\end{itemize}
\end{lemma}
Using Lemma \ref{lem:dist} a), the parent-child relationships for a set of observed nodes $\mathcal O$ can be identified. Similarly, using Lemma \ref{lem:dist} b), sibling groups can be identified. 

The formal description of RG is given in Algorithm \ref{alg:rg}. The input of RG is a set of observed nodes $\mathcal O\subset\mathcal V$ and the additive distance $d(a,b)$ for all $a,b\in\mathcal O$. Now, we discuss the working of RG steps through an illustrative example given in Fig. \ref{fig:rg}, where green nodes represent $\mathcal O$. 
First, RG finds groups of siblings and their parents using Lemma \ref{lem:dist}, as shown in Fig. \ref{fig:rg2}. Edges are added between all identified parent-child pairs.
For identified siblings without an observed parent, a parent node is inserted and connected to its children, as shown in Fig. \ref{fig:rg3}. $d(\cdot,\cdot)$ is updated for the newly added parents using the fact that distances are additive. For siblings $a,b\in\mathcal O$ and their newly added parent $h$, the distances $d(a,h)$ and $d(c,h)$ for any $c \in \mathcal{O}$ are calculated by
\begin{align}
d(a,h)&=\frac{1}{2}(d(a,b)+\Phi_{abc}), \text{~any $c\in\mathcal O$}\label{eq:dist1}\\
d(c,h)&=d(a,c)-d(a,h)\label{eq:dist2}
\end{align}
Finally, RG updates $\mathcal O$ with newly added parents and nodes without established parent or child relations illustrated by green nodes in Fig. \ref{fig:rg3}. The process is iterated, and new edges are added unless $|\mathcal O|\le 2$, which applies when an edge can be added to remaining vertices or when a single vertex is left.
Fig. \ref{fig:rg4}-\ref{fig:rg6} illustrates iterations of the RG after the first one. 

\begin{algorithm}\caption{Recursive Grouping Algorithm ($\mathtt{RG})$}\label{alg:rg}
\begin{algorithmic}[1]
\State Input: $\mathcal O$, $\{d(a,b):a,b\in\mathcal O\}$
\State Output: $(\mathcal V,\mathcal E)$, $\{d(a,b):a,b\in\mathcal V\}$
\State Initialization: $\mathcal V=\mathcal O, \mathcal E=\emptyset$
\While{$|\mathcal O|>2$}
\State $\mathcal O_{NEW}\leftarrow\emptyset$.
\State Compute $\Phi_{abc}=d(a,c)-d(b,c)$ for all $a,b,c\in \mathcal O$.
\State Find a coarsest partition $\Pi$ of $\mathcal O$ such that any two nodes in $S\in\Pi$ are either leaves and sibling, or a parent and a leaf child.\footnotemark
\For{$S\in\Pi$}
\If{$|S|=1$}
\State $\mathcal O_{NEW}\leftarrow \mathcal O_{NEW}\cup S$.
\ElsIf{a parent $p_S\in S$ exists}
\State $\mathcal E\leftarrow\mathcal E\cup\big\{(p_Sa):a\in S\setminus\{p_S\}\big\}$
\State $\mathcal O_{NEW}\leftarrow \mathcal O_{NEW}\cup\{p_S\}$
\Else
\State Add a parent $h_S$ of $S$ as follows
\State $\mathcal V\leftarrow\mathcal V\cup\{h_S\}$
\State $\mathcal E\leftarrow\mathcal E\cup\big\{(h_Sa):a\in S\big\}$
\State $\mathcal O_{NEW}\leftarrow \mathcal O_{NEW}\cup\{h_S\}$
\EndIf
\EndFor
\State Update $d(\cdot,\cdot)$ for $\mathcal O_{NEW}$ using Eqs.~(\ref{eq:dist1}, \ref{eq:dist2}).
\State $\mathcal O\leftarrow \mathcal O_{NEW}$.
\EndWhile
\If{$|\mathcal O|=2$}
\State $\mathcal E\leftarrow\mathcal{E}\cup\{(ab):a,b\in\mathcal O, a\ne b\}$
\EndIf
\end{algorithmic}
\end{algorithm}
\footnotetext{$\Pi$ is a coarsest partition if for any $\Pi^\prime$ and for any $S^\prime\in\Pi^\prime$, there exists $S\in\Pi$ such that $S^\prime\subset S$. The coarsest partition $\Pi$ in Algorithm \ref{alg:rg} represents a collection of sets of siblings and their parent.}
For topology estimation in radial grids, we propose the following two stage algorithm with missing modes:
\begin{itemize}
 \item[1.]For all $a,b\in\mathcal O$, calculate ${d}_r(a,b)$ and ${d}_x(a,b)$ using Eqs.~(\ref{eq:dist3}, \ref{eq:dist4}) and second order moments.
 \item[2.]Recover missing nodes and edges using RG.
\end{itemize}
The formal statement of the algorithm is presented in Algorithm \ref{alg:main}. Note that by learning the impedances based distances, the impedance of each operational edge is jointly estimated along with the topology. This is possible due to the availability of injection samples that enable computation of the right side of Eqs.~(\ref{eq:dist3}). The previous Algorithm \ref{alg:learningdeep} used only injection statistics, which are not sufficient for impedance estimation.
Next, we briefly some extensions of Algorithm \ref{alg:main}.\\
\noindent{\bf Learning with multiple substations:} In the setting where multiple substations may exist, with each powering a subset of the buses in a radial topology, we need to first separate the observed buses into groups, one group per substation. As shown in \cite{dekatcns}, bus voltages in distribution grids under different substations are uncorrelated. Thus, the separation of buses can be done by looking at the correlation in their observed voltage magnitudes alone before running Algorithm \ref{alg:main}.\\
\noindent{\bf Learning without Assumption \ref{asm:deg3}}:
Note that RG estimates the additive distance to identify sibling nodes and then recovers their parent. If some internal node $b$ has degree 2 (its child $c$ has no sibling), then $b$ cannot be identified using Lemma \ref{lem:dist} b). Instead, if $b$'s parent $a$ has a degree $>2$, then $b$'s child $c$ will get connected to $a$. In other words, RG outputs a topology without degree 2 nodes by adding edges between their parent and child. This reduced graph is exactly the Kron-reduced model \cite{kron} derived by removing degree $2$ nodes from the grid graph. Note that the estimated line impedance of discovered edge $(ac)$ will be the sums of impedances of the two edges $(ab)$ and $(bc)$ that connect $a$ and $c$ in the original graph, where missing node $b$ has degree $2$. This is indeed the true impedance in the Kron reduced graph when missing nodes of degree $2$ are removed. Thus, our algorithm preserves the impedance in the reduced graph. Similarly, if some leaf nodes are not observed, their resistive/reactive distances to other leaves are not computed. However, it does not affect the topology learning in the remaining graph without the missing leaves.\\
\noindent\textbf{Recovering unobserved internal injection and voltages:}
Once the topology and line impedances are correctly estimated, one can also recover unobserved time-stamped voltage magnitude/phase and complex power injection samples of missing nodes by solving linear equations at each time step: LC-PF Eq.~\eqref{eq:lcpf2}. Moreover, one can also compute the second order statistics (variances of injections) of missing nodes directly using similar relations that relate the covariances Eq.~\eqref{eq:expand}. Since the recovery is performed using time-stamped samples,
it can be easily extended to the finite sample case (See Section \ref{sec:finite} and Section \ref{sec:complexity} for more information).
\begin{algorithm}[t]\caption{Topology/Impedance Learning Algorithm with Voltage and Power Samples}\label{alg:main}
\begin{algorithmic}[1]
\State Input: $\mathcal O$, $\{\mathbb{E}[v_ap_b],\mathbb{E}[v_aq_b],\mathbb{E}[p_a^2],\mathbb{E}[q_a^2],\mathbb{E}[p_aq_a]:a,b\in\mathcal O\}$
\State Output: $(\mathcal V,\mathcal E)$, $\{r_{ab},x_{ab}:(ab)\in\mathcal E\}$
\For{$a,b\in\mathcal O$}
\State $\begin{bmatrix}
 H_{1/r}^{-1}(a,b) \\ H_{1/x}^{-1}(a,b)
 \end{bmatrix}\leftarrow
 \begin{bmatrix}
 \mathbb{E}[p_b^2] & \mathbb{E}[p_bq_b] \\ \mathbb{E}[p_bq_b] & \mathbb{E}[q_b^2]
 \end{bmatrix}^{-1}\begin{bmatrix}
 \mathbb{E}[v_ap_b] \\ \mathbb{E}[v_aq_b]
 \end{bmatrix}$
 \vspace{0.03in}
\EndFor
\For{$a,b\in\mathcal O$}
\State $d_r(a,b)\leftarrow H_{1/r}^{-1}(a,a)+H_{1/r}^{-1}(b,b)-2H_{1/r}^{-1}(a,b)$
\State $d_x(a,b)\leftarrow H_{1/x}^{-1}(a,a)+H_{1/x}^{-1}(b,b)-2H_{1/x}^{-1}(a,b)$
\EndFor
\State $(\mathcal V,\mathcal E),\{d_r(a,b):a,b\in\mathcal V\}\leftarrow\mathtt{RG}(\mathcal O,\{d_r(a,b):a,b\in\mathcal O\})$
\For{$(ab)\in\mathcal E$}
\State $r_{ab}\leftarrow d_r(a,b)$,
 $x_{ab}\leftarrow d_x(a,b)$ where $d_x(a,b)$ is obtained using $(\mathcal V, \mathcal E)$
\EndFor
\end{algorithmic}
\end{algorithm}
\vspace{-0.075in}
\subsection{Recursive Grouping with Finite Samples}\label{sec:finite}
In a practical scenario, due to the finite number of injection and voltage samples, one can compute only the approximated value $\widehat d_r$ of $d_r$ rather than the exact value. In other words, the variance of the distance is nonzero. To account for it, we allow some tolerance $\varepsilon$ for finding parent-child and sibling relationships in Lemma \ref{lem:dist}. In addition, to test the relationship of $a,b$, we only use nodes that are close enough to both $a$ and $b$, i.e., nodes in $\mathcal K_{ab}$
where $\mathcal K_{ab}$ satisfies
$$\mathcal K_{ab}=\{c\in\mathcal O\setminus\{a,b\}:\widehat d_r(a,c),\widehat d_r(b,c)<\tau\}$$
for some constant $\tau$. Let us now present rules which guide the relationships of nodes using samples.
\begin{itemize}
\item[a)] Set $a$ as a parent of $b$ if $|\widehat d_r(a,b)-\widehat\Phi_{abc}|\le\varepsilon~\forall c\in\mathcal K_{ab}$.
\vspace{0.05in}
\item[b)] Set $a,b$ as siblings if $\underset{{c\in \mathcal K_{ab}}}{\max}\widehat\Phi_{abc}-\underset{{c\in \mathcal K_{ab}}}{\min}\widehat\Phi_{abc}\le\varepsilon.$ 
\end{itemize}
Update of the distance is done in a similar manner where we use averaging to mitigate the variability from finite sample sizes. For $a\in\mathcal O$ and its newly added parent $h$, we set
\begin{align*}
&\widehat d_r(a,h)=\\
&~\frac{1}{2(|\mathcal C(h)|-1)}\sum_{b\in\mathcal C(h)\setminus a}\left(\widehat d_r(a,b)+\frac{1}{|\mathcal K_{ab}|}\sum_{c\in \mathcal K_{ab}}\widehat\Phi_{abc}\right)
\end{align*}
where $\mathcal C(h)$ denotes the children set of $h$. Likewise, for $c\notin\mathcal C(h)$, we set
\begin{align*}
\widehat d_r(c,h)=\frac{1}{|\mathcal C(h)|}\sum_{a\in\mathcal C(h)}\left(\widehat d_r(a,c)-\widehat d_r(a,h)\right).
\end{align*}

\vspace{-0.075in}
\subsection{Recursive Grouping with Non-linear Power flows}
In the last section, we introduce the consistent topology and impedance learning algorithm under the LC-PF model.
However, as we are interested in simulations over samples generated by non-linear ac power flow models, there is some limitation for directly applying Algorithm \ref{alg:main} into real examples. In this section, we address these bottlenecks and propose a simple variant of Algorithm \ref{alg:main} for practical implementations.

First, Algorithm \ref{alg:rg} is very sensitive to the tolerances used for finite sample lengths. If the tolerance is too small, the algorithm outputs an error as it cannot find sibling relationships. In contrast, if the tolerance is too large, the algorithm outputs a loose topology with a small number of missing nodes, which in the worst case can result in a star topology. Second, even with an infinite number of samples, since the real model is not linear, the approximated distance does not converge to the real distance. 
This causes a serious problem in large grids as different distances may have different errors that may not be handled by a fixed tolerance. 
Third, the algorithm does not utilize $d_r$ and $d_x$ at once. 
To resolve these issues, we dynamically vary the tolerance $\varepsilon$ in our experiments as follows:
\begin{itemize}
 \item[1.] In Algorithm \ref{alg:rg}, if no parent-child relationship is updated, set $d\leftarrow d_x$ instead of $d_r$ and iterate.
 If the parent-child relationship is updated, set $d\leftarrow d_r$.
 \item[2.] If the algorithm does not find the parent-child relationship after Step 1, increase the tolerance value $\varepsilon\leftarrow \alpha\varepsilon$ ($\alpha>1$) and set $d\leftarrow d_r$.
 If the parent-child relationship is updated, reset $\varepsilon$ to the initial input value. Otherwise, go to Step 1. 
\end{itemize}
Note that this procedure uses both resistive and reactive distances to determine edges. If both fail due to a small tolerance value, the algorithm increases the tolerance to find the appropriate relationships. The possibility to consider several values of tolerance can also help in handling the non-linearity of ac power flow models. 
We note that this modified algorithm is at least good as Algorithm \ref{alg:main} under the LC-PF model.
\vspace{-0.075in}
\subsection{Sample and Computational Complexity}\label{sec:complexity}
In this section, we show that Algorithm \ref{alg:main} has $O(d|\mathcal V|^3)$ computational complexity and under mild assumptions, it has $O(|\mathcal V|\log|\mathcal V|)$ sample complexity where $d$ denotes the depth of the grid.
In the computational complexity, $|\mathcal V|^3$ follows from computing $\Phi_{abc}$ in RG and $d$ follows from the iterations of RG. The sample complexity result is stated in the following theorem where its proof is presented in Section \ref{sec:pfthm:main}.
\begin{theorem}\label{thm:main}
Suppose that a radial grid $(\mathcal V,\mathcal E)$ has a constant depth.
Under Assumptions \ref{asm:deg3}-\ref{asm:bound} and assuming the LC-PF model, if line impedances are constantly upper and lower bounded, nodal power injections are zero mean sub-Gaussian with constantly bounded sub-Gaussian parameters, and the number of samples is greater than $C|\mathcal V|\log(|\mathcal V|/\eta)$ for some constant $C$, then there exist $\varepsilon,\tau>0$ such that Algorithm \ref{alg:main} recovers the true topology with probability at least $1-\eta$.
\end{theorem}

We note that the sub-Gaussian distribution in Theorem \ref{thm:main} is a broad class of light-tail distributions, including the Gaussian distribution, defined as below.
\begin{definition}
A zero mean random variable $X$ is sub-Gaussian if there exists a constant $K\ge 0$ such that $\mathbb{P}(|X|>t)\le e^{1-t^2/K^2}$ for all $t\ge 0$.
\end{definition}

\vspace{-0.075in}
\subsection{Proof of Theorem \ref{thm:main}}\label{sec:pfthm:main}
We first provide the following key lemma that the estimated resistance distances are uniformly bounded from the true distances if $C|\mathcal V|\log(|\mathcal V|/\eta)$ samples are given.
\begin{lemma}\label{lem:uniformbound}
Under assumptions in Theorem \ref{thm:main}, for any constant $\delta>0$, there exists a constant $C>0$ such that if the number of samples is greater than $C|\mathcal V|\log(|\mathcal V|/{\eta})$, then,
$|\widehat d_r(a,b)-d_r(a,b)|\le\delta$ for all $a,b\in\mathcal O$
with probability at least $1-\eta$.
\end{lemma}
If $|\widehat d_r(a,b)-d_r(a,b)|\le\delta$ holds for a sufficiently small constant $\delta$, one can observe that RG recovers the true topology for some $\varepsilon,\tau$ due to the constant depth and constantly lower bounded line impedances, i.e., Lemma \ref{lem:uniformbound} implies Theorem \ref{thm:main}. 
\begin{proof}[Proof of Lemma \ref{lem:uniformbound}]\label{sec:pflem:uniformbound}
We show that the empirical expectations $\widehat{\mathbb{E}}[v_ap_b]$, $\widehat{\mathbb{E}}[v_aq_b]$, $\widehat{\mathbb{E}}[p_a^2]$, $\widehat{\mathbb{E}}[q_a^2]$ and $\widehat{\mathbb{E}}[p_aq_a]$ are close enough to its true expectations so that the result of Lemma \ref{lem:uniformbound} holds.
For bounding errors, we first define the error event $$\mathcal E_{v_ap_b}(\zeta):=\{|\mathbb{\widehat E}[v_ap_b]-\mathbb{E}[v_ap_b]|\ge\zeta\}$$
where $\mathcal E_{v_aq_b}(\zeta)$, $\mathcal E_{p_a^2}(\zeta)$, $\mathcal E_{q_a^2}(\zeta)$, $\mathcal E_{p_aq_b}(\zeta)$ are also defined in a similar manner.
To bound the probability of error events, we introduce the following key lemmas.
\begin{lemma}\label{lem:subgaus}
Under assumptions in Theorem \ref{thm:main}, the following inequalities hold for some constants $\alpha,M>0$: for $|\lambda|\le{M}$
$$\mathbb{E}[e^{\lambda (p_a^2-\mathbb{E}[p_a^2])}],\mathbb{E}[e^{\lambda (q_a^2-\mathbb{E}[q_a^2])}],\mathbb{E}[e^{\lambda (p_aq_q-\mathbb{E}[p_aq_a])}]\le e^{\lambda^2\alpha^2}$$
and for $|\lambda|\le\frac{M}{\sqrt{|\mathcal V|}}$
$$\mathbb{E}[e^{\lambda (v_ap_b-\mathbb{E}[v_ap_b])}],\mathbb{E}[e^{\lambda (v_aq_b-\mathbb{E}[v_aq_b])}]\le e^{\lambda^2|\mathcal V|\alpha^2}.
$$
\end{lemma}
\begin{lemma}\label{lem:bernstein}
Let $X_1,\dots,X_n$ are independent random variables satisfying 
$$\mathbb{E}[e^{\lambda X_i}]\le e^{\lambda^2\sigma^2}\quad\text{for $|\lambda|\le B$}.$$
Then, the following inequality holds:
$$\mathbb{P}\left(\left|\sum_{i=1}^nX_i\right|\ge t\right)\le 2\exp\left(-\min\left(\frac{Bt}{2},\frac{t^2}{4\sigma^2n}\right)\right).$$
\end{lemma}
The proofs of Lemma \ref{lem:subgaus} and Lemma \ref{lem:bernstein} are presented in Appendix \ref{sec:pflem:subgaus} and Appendix \ref{sec:pflem:bernstein} respectively.
Lemma \ref{lem:subgaus} and Lemma \ref{lem:bernstein} directly lead us to obtain the following inequality for any constant $\zeta$:
\begin{align}
\mathbb{P}(\mathcal E_{v_ap_b}(\zeta))&\le 2\exp\left(-\min\left(\frac{M\zeta n}{2\sqrt{|\mathcal V|}},\frac{\zeta^2 n}{4|\mathcal V|\alpha^2}\right)\right)\le\frac{\eta}{4|\mathcal V|^2}\label{eq:errbd1}
\end{align}
for $n\ge D_1(\zeta)|\mathcal V|\log(|\mathcal V|/\eta)$ and some constant $D_1(\zeta)$
where $n$ is the number of samples. The same inequality holds for $v_aq_b$. Similarly, for any constant $\zeta$, the inequality below holds
\begin{align}
\mathbb{P}(\mathcal E_{p_aq_a}(\zeta))&\le 2\exp\left(-\min\left(\frac{M\zeta n}{2},\frac{\zeta^2 n}{4\alpha^2}\right)\right)\le\frac{\eta}{6|\mathcal V|}\label{eq:errbd2}
\end{align}
for $n\ge D_2(\zeta)\log(|\mathcal V|/\eta)$ for some constant $D_2(\zeta)$. One can observe that same inequality holds for $p_a^2$ and $q_a^2$.
Now, we define the global error event
\begin{align*}
\mathcal E(\zeta):=&\Bigg(\bigcup_{a,b\in\mathcal V}\mathcal E_{v_ap_b}(\zeta)\cup\mathcal E_{v_aq_b}(\zeta)\Bigg)\nonumber\\
&\cup\Bigg(\bigcup_{a\in\mathcal V}\mathcal E_{p_a^2}(\zeta)\cup\mathcal E_{p_aq_a}(\zeta)\cup\mathcal E_{q_a^2}(\zeta)\Bigg).
\end{align*}
Using two inequalities Eqs.~(\ref{eq:errbd1}, \ref{eq:errbd2}), we apply the union bound to bound the error probability as follows:
\begin{align*}
&\mathbb{P}\left(\mathcal E(\zeta)\right)\le 2|\mathcal V|^2\times\frac{\eta}{4|\mathcal V|^2}+3|\mathcal V|\times\frac{\eta}{6|\mathcal V|}=\eta.
\end{align*} 
Using the above union bound, we choose a small enough constant $\zeta$ and its corresponding number of samples $C|\mathcal V|\log|\mathcal V|$, where $C=\max(D_1(\zeta),D_2(\zeta))$, so that if $\mathcal E(\zeta)$ does not occur, then $|d_r(a,b)-\widehat{d}_r(a,b)|\le\delta$ where
$$\widehat{d}_r(a,b)=\widehat{H}_{1/r}^{-1}(a,a)+\widehat{H}_{1/r}^{-1}(b,b)-2\widehat{H}_{1/r}(a,b)$$
and
\begin{align*}
\begin{bmatrix}\widehat{H}_{1/r}^{-1}(a,b)\\
\widehat{H}_{1/x}^{-1}(a,b)\end{bmatrix}&=\frac1{\widehat{\mathbb{E}}[p_b^2]\widehat{\mathbb{E}}[q_b^2]-(\widehat{\mathbb{E}}[p_bq_b])^2}\\
&\quad~~\times\begin{bmatrix}
\widehat{\mathbb{E}}[q_b^2] & -\widehat{\mathbb{E}}[p_bq_b]\\
-\widehat{\mathbb{E}}[p_bq_b] & \widehat{\mathbb{E}}[p_b^2]\end{bmatrix}\begin{bmatrix}\widehat{\mathbb{E}}[v_ap_b]\\
\widehat{\mathbb{E}}[v_aq_b]\end{bmatrix}
\end{align*}
which are from Eqs. (\ref{eq:dist3}, \ref{eq:dist4}).
Such a constant $\zeta$ always exists due to Assumption \ref{asm:bound}. 
This completes the proof of Lemma~\ref{lem:uniformbound}.
\end{proof}

\begin{figure*}[t]
\centering
 \begin{subfigure}{0.32\textwidth}
 \includegraphics[width=0.49\textwidth]{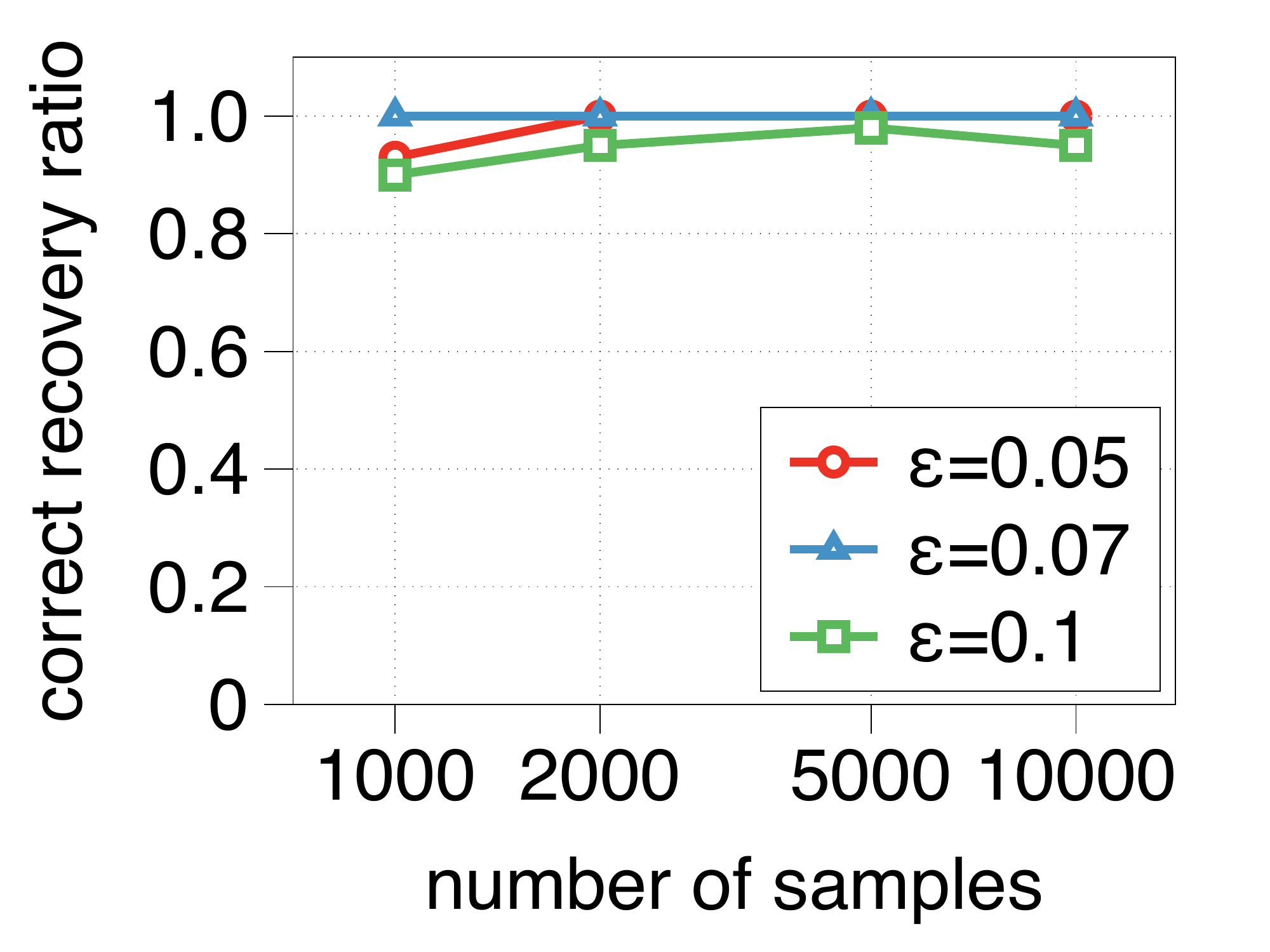}
 \includegraphics[width=0.49\textwidth]{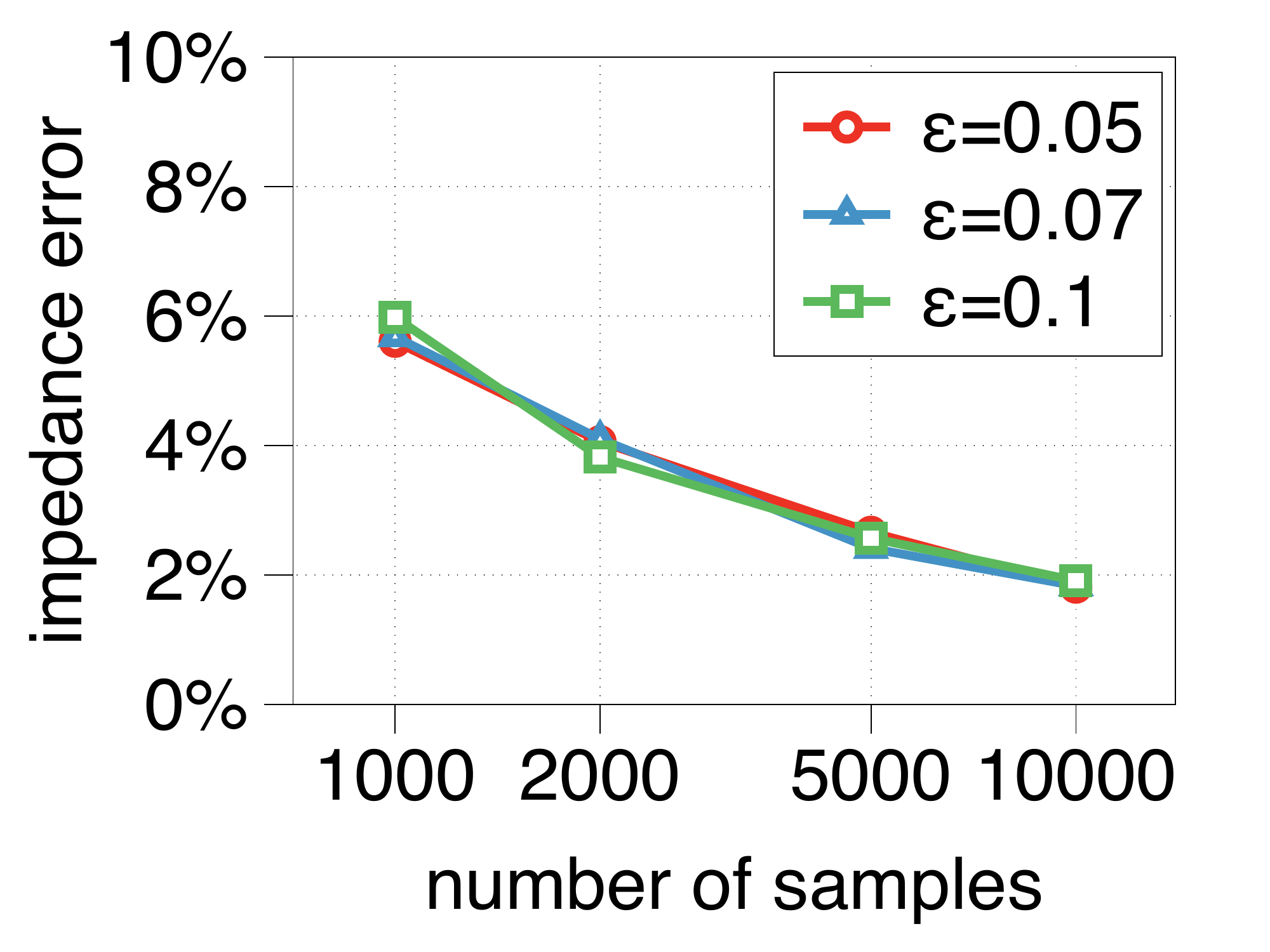}
 \caption{10 vertices}
 \label{fig:syn10}
 \end{subfigure}
 ~
 \begin{subfigure}{0.32\textwidth}
 \includegraphics[width=0.49\textwidth]{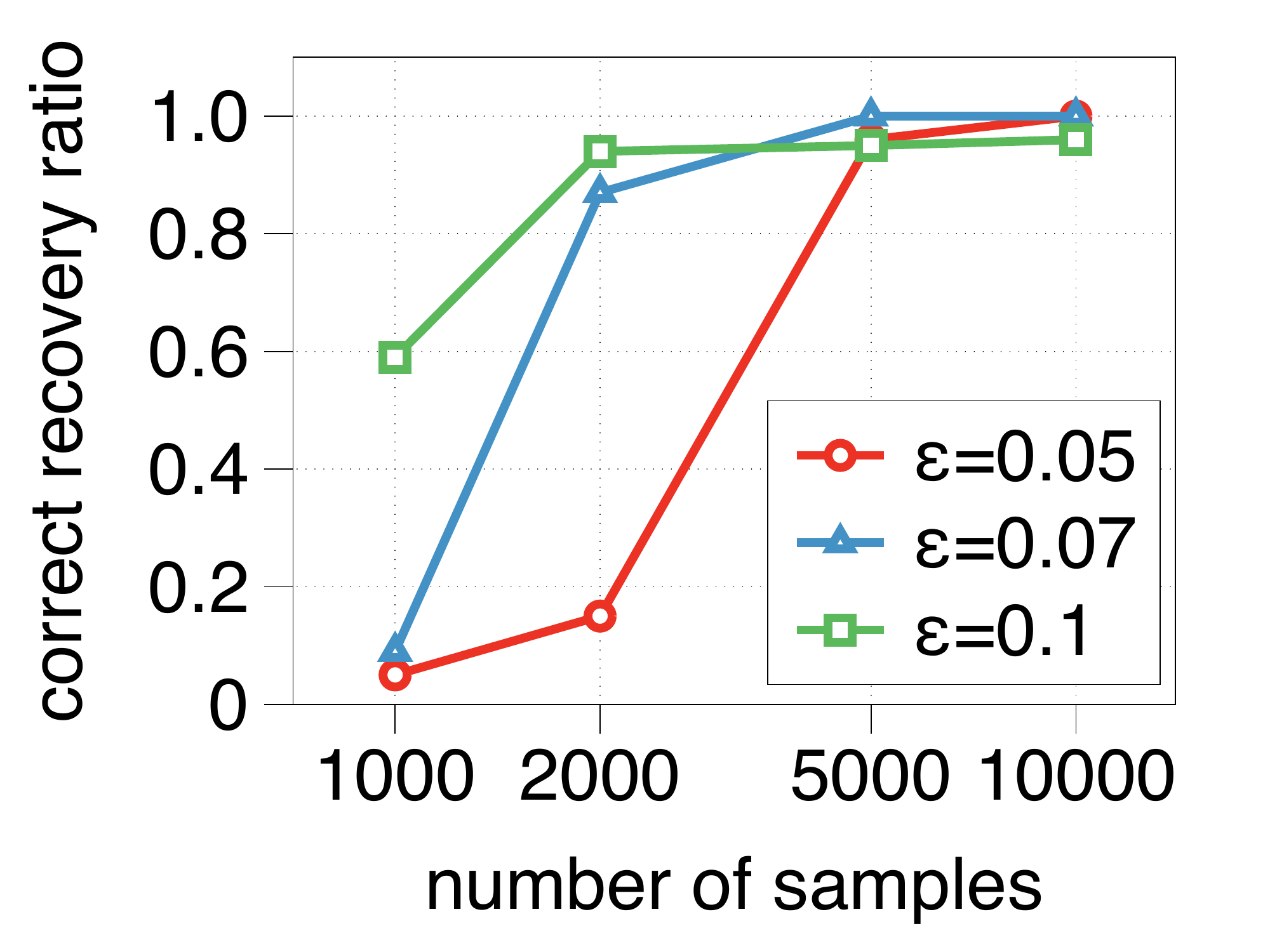}
 \includegraphics[width=0.49\textwidth]{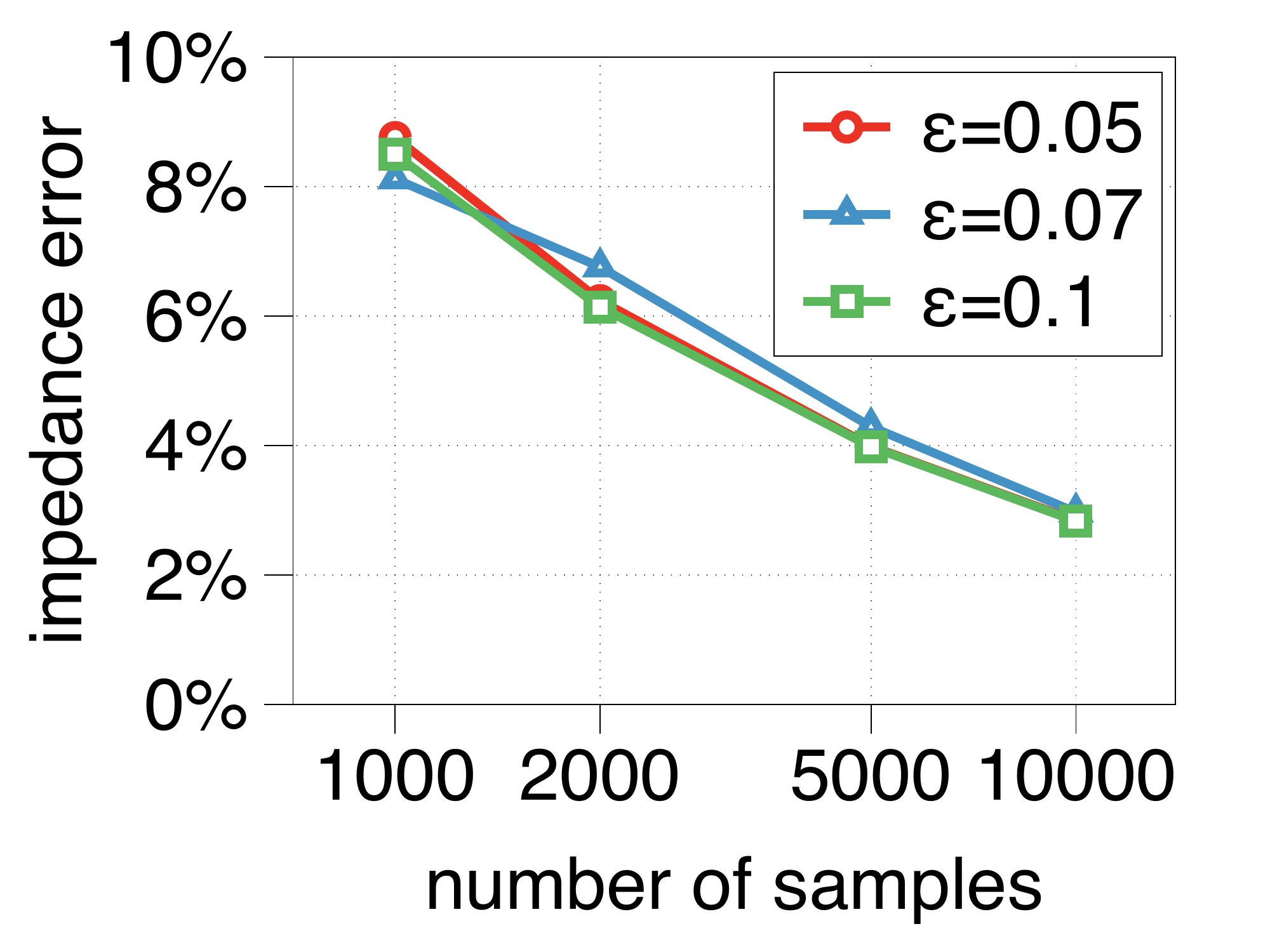}
 \caption{50 vertices}
 \label{fig:syn50}
 \end{subfigure}
 ~
 \begin{subfigure}{0.32\textwidth}
 \includegraphics[width=0.49\textwidth]{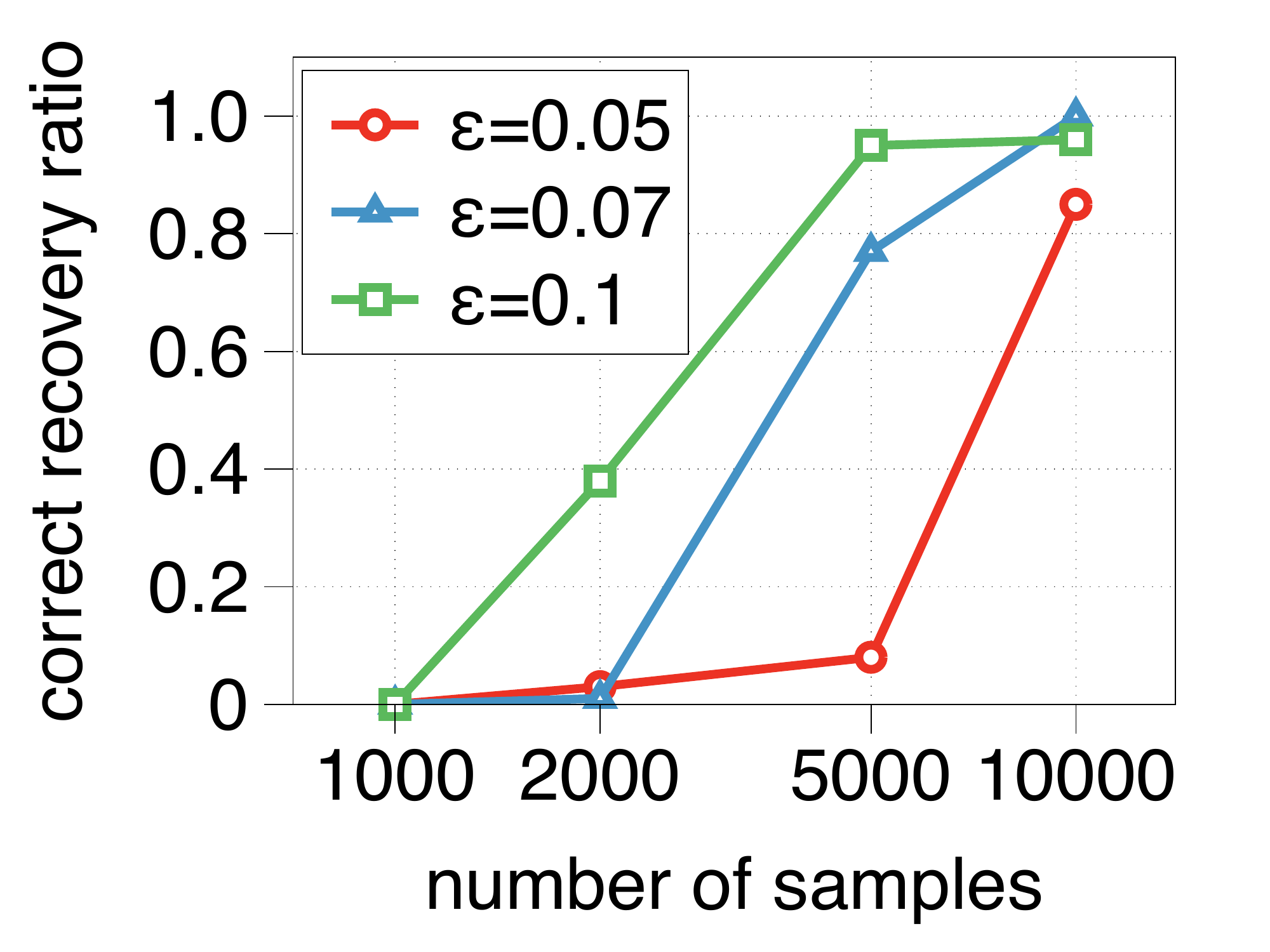}
 \includegraphics[width=0.49\textwidth]{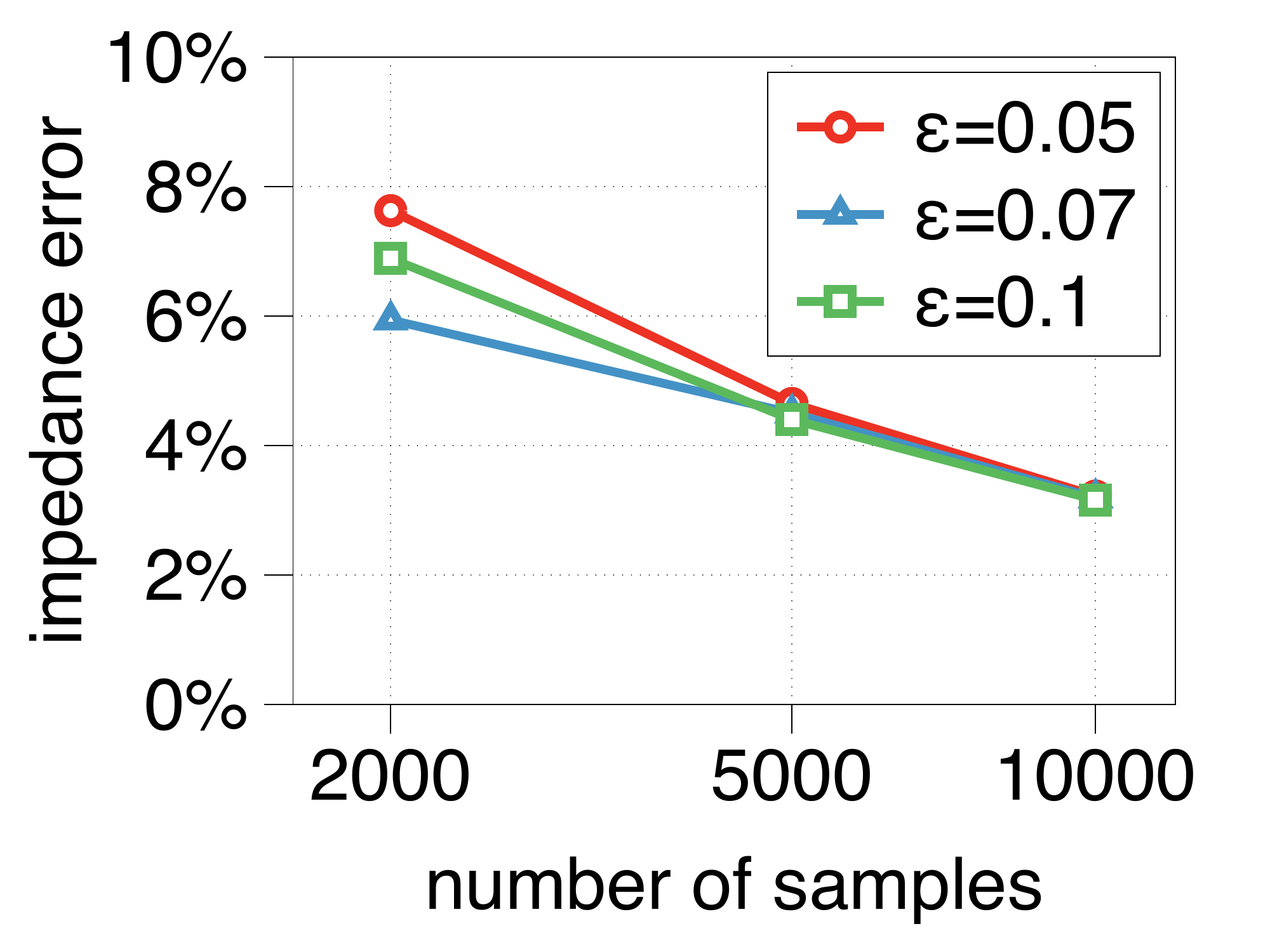}
 \caption{100 vertices}
 \label{fig:syn100}
 \end{subfigure}
\caption{Experimental results for synthetic grids with 10, 50, and 100 vertices averaged over 100 random radial grids, measured in terms of the accuracy of topology recovery, and the impedance error measured as $\frac{1}{2|\mathcal{E}|}\sum_{(ab)\in\mathcal E}\frac{|r_{ab}-\hat r_{ab}|}{|r_{ab}|}+\frac{|x_{ab}-\hat x_{ab}|}{|x_{ab}|}$ in the correctly recovered topologies. $\varepsilon$ denotes the threshold in Algorithm \ref{alg:main}, as described in Section \ref{sec:finite}} 
\label{fig:syn}
\end{figure*}
\section{Experiments}\label{sec:experiments}
In this section, we present experimental results of Algorithm \ref{alg:learningdeep} and Algorithm \ref{alg:main} on custom grids and IEEE test cases for both LC-PF, non-linear ac power flow, and real-world samples.

\noindent{\bf Custom Examples with LC-PF samples:}
We first run simulations on randomly designed grids with voltages generated by the LC-PF model. Due to the space constraint, we only simulate Algorithm \ref{alg:main} for the random grids and postpone the discussion of Algorithm \ref{alg:learningdeep} to IEEE cases.
In each simulation, we construct a random radial grid with maximum degree 5, and line resistances and reactances independently sampled from the uniform distribution over $[.1,.2]$. We sample the complex power injections from the independent normal distribution and produce nodal voltage using LC-PF Eq.~(\ref{eq:lcpf2}).

Under this setting, we run simulations by changing the number of vertices from 10 to 100, the number of samples from 1000 to 10000, and changing tolerance $\varepsilon$ in the algorithm with fixed $\tau=1$. To quantify the performance, we measure errors in the recovered topology and estimated impedances, averaged over 100 random radial grids. 
The results are summarized in Fig. \ref{fig:syn} where the impedance error is defined as $\frac{1}{2|\mathcal{E}|}\sum_{(ab)\in\mathcal E}\frac{|r_{ab}-\hat r_{ab}|}{|r_{ab}|}+\frac{|x_{ab}-\hat x_{ab}|}{|x_{ab}|}$. One can observe that our algorithm recovers line impedances with a small error even in the demanding case of 1000 samples. We also observe that larger $\varepsilon$ results in higher accuracy for a small number of samples, but it becomes less accurate for a large number of samples (compare $\varepsilon=0.1$ to $\varepsilon=0.07$).
However, if the threshold is too small ($\varepsilon=0.05$), the algorithm performance decreases for all sample sizes. Note that similar results (thus not shown) are derived when changing the variance of the complex power injections.

\begin{figure}[ht]
\centering
 \begin{subfigure}[b]{0.14\textwidth}
 \includegraphics[width=\textwidth]{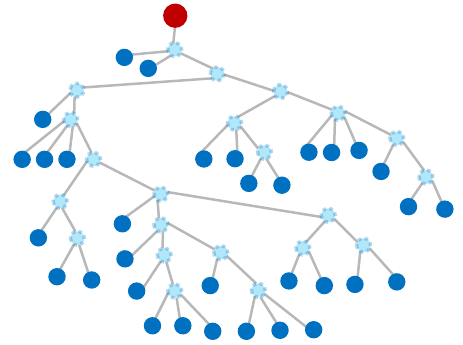}
 \caption{}
 \label{fig:ieee123grid}
 \end{subfigure}
 ~
 \begin{subfigure}[b]{0.14\textwidth}
 \includegraphics[width=\textwidth]{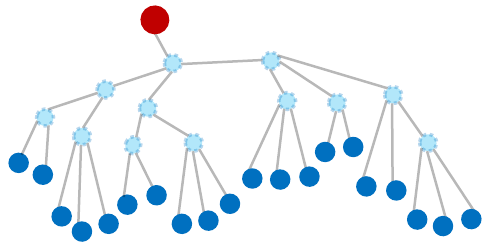}
 \caption{}
 \label{fig:33bwgrid}
 \end{subfigure}
 \caption{Illustrations of (a) $56$ bus distribution grid with $33$ leaf nodes, $22$ internal nodes 
 (b) $33$ bus distribution grid with $20$ leaf nodes, $12$ internal nodes.
 The substation is colored red.}
\end{figure}

\begin{figure*}[t]
\centering
 \begin{subfigure}{0.185\textwidth}
 \includegraphics[width=\textwidth]{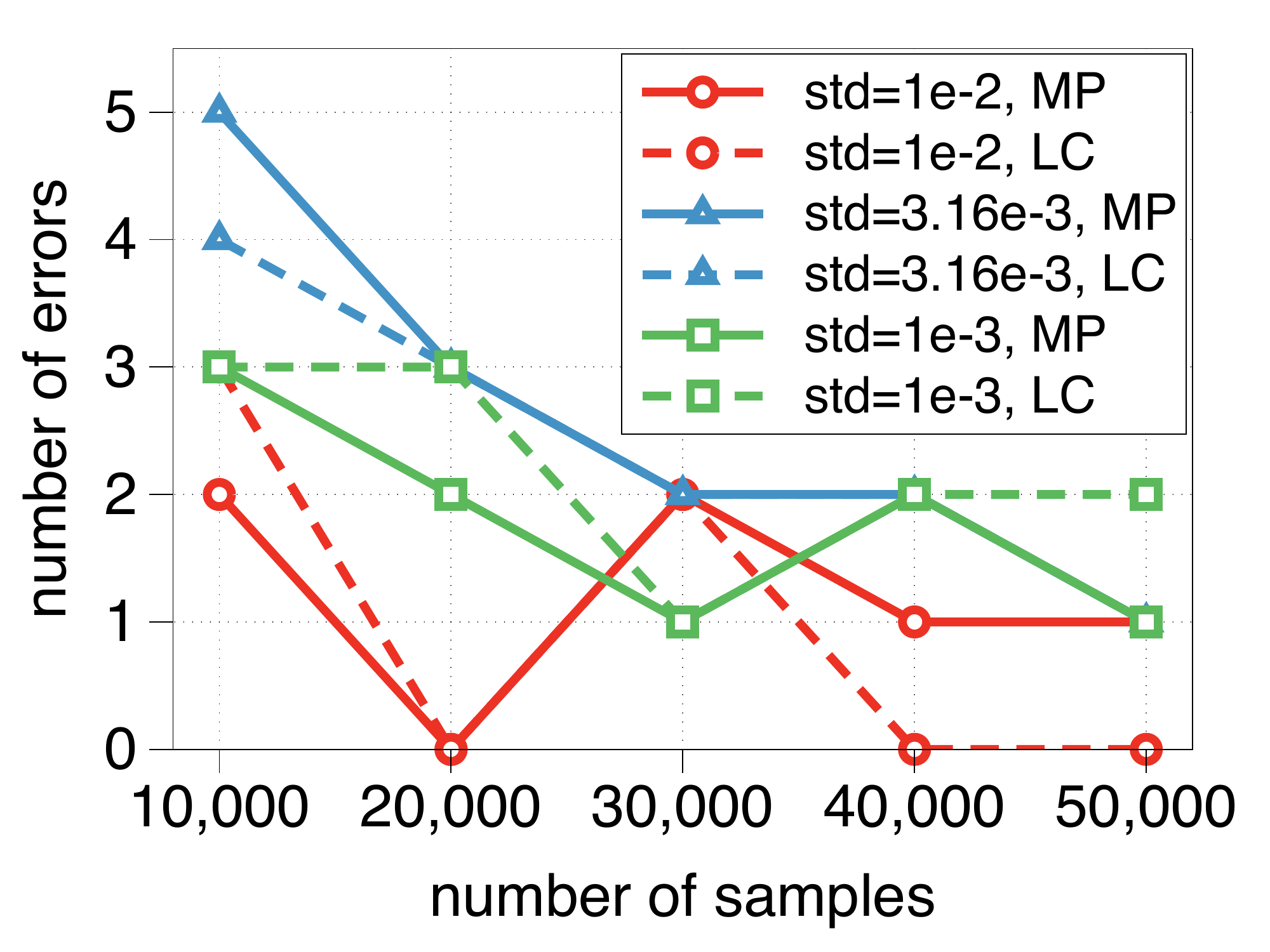}
 \caption{}
 \label{fig:ieee123var}
 \end{subfigure}
 ~
 \begin{subfigure}{0.185\textwidth}
 \includegraphics[width=\textwidth]{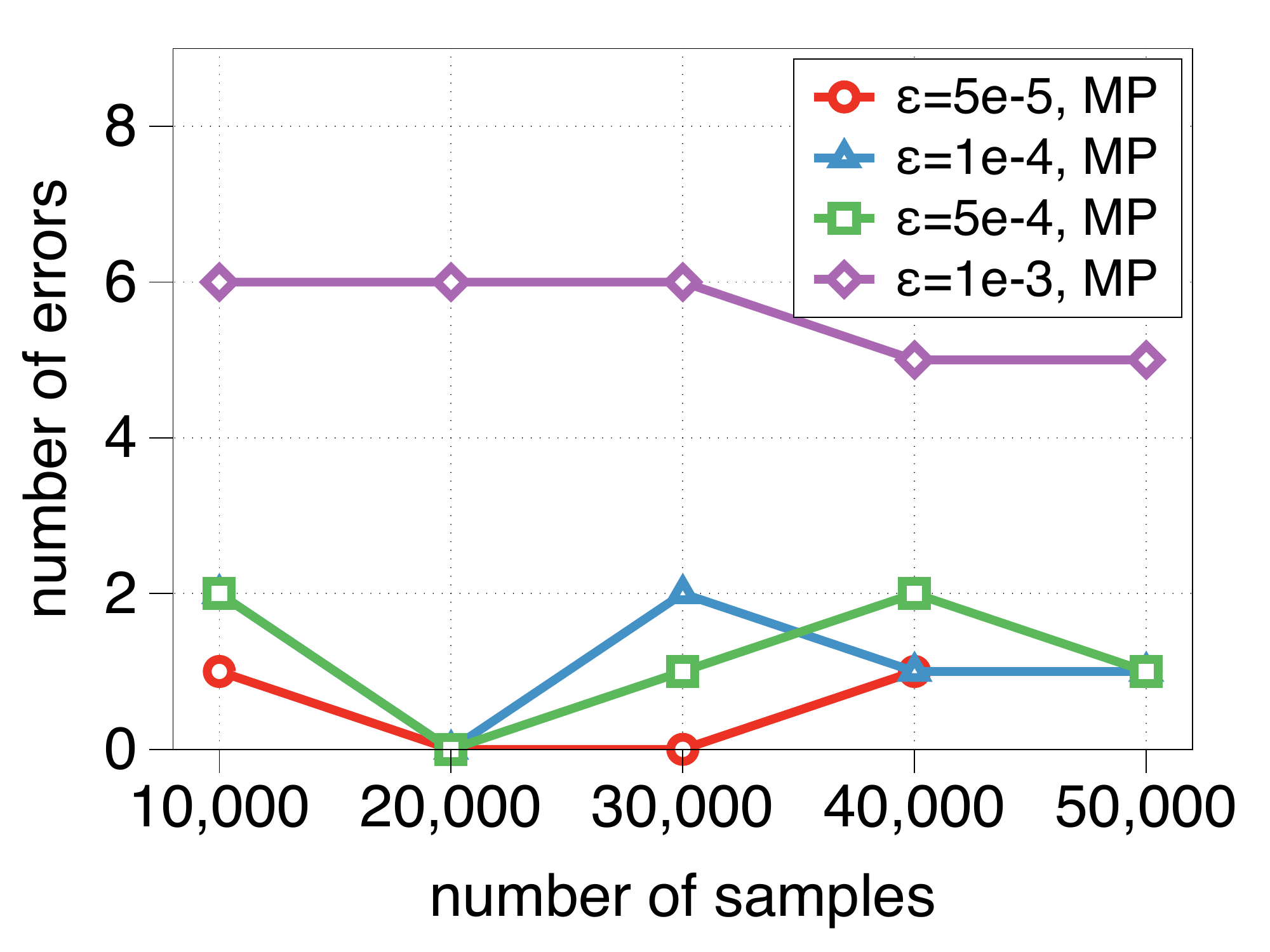}
 \caption{}
 \label{fig:ieee123eps}
 \end{subfigure}
 ~
 \begin{subfigure}{0.185\textwidth}
 \includegraphics[width=\textwidth]{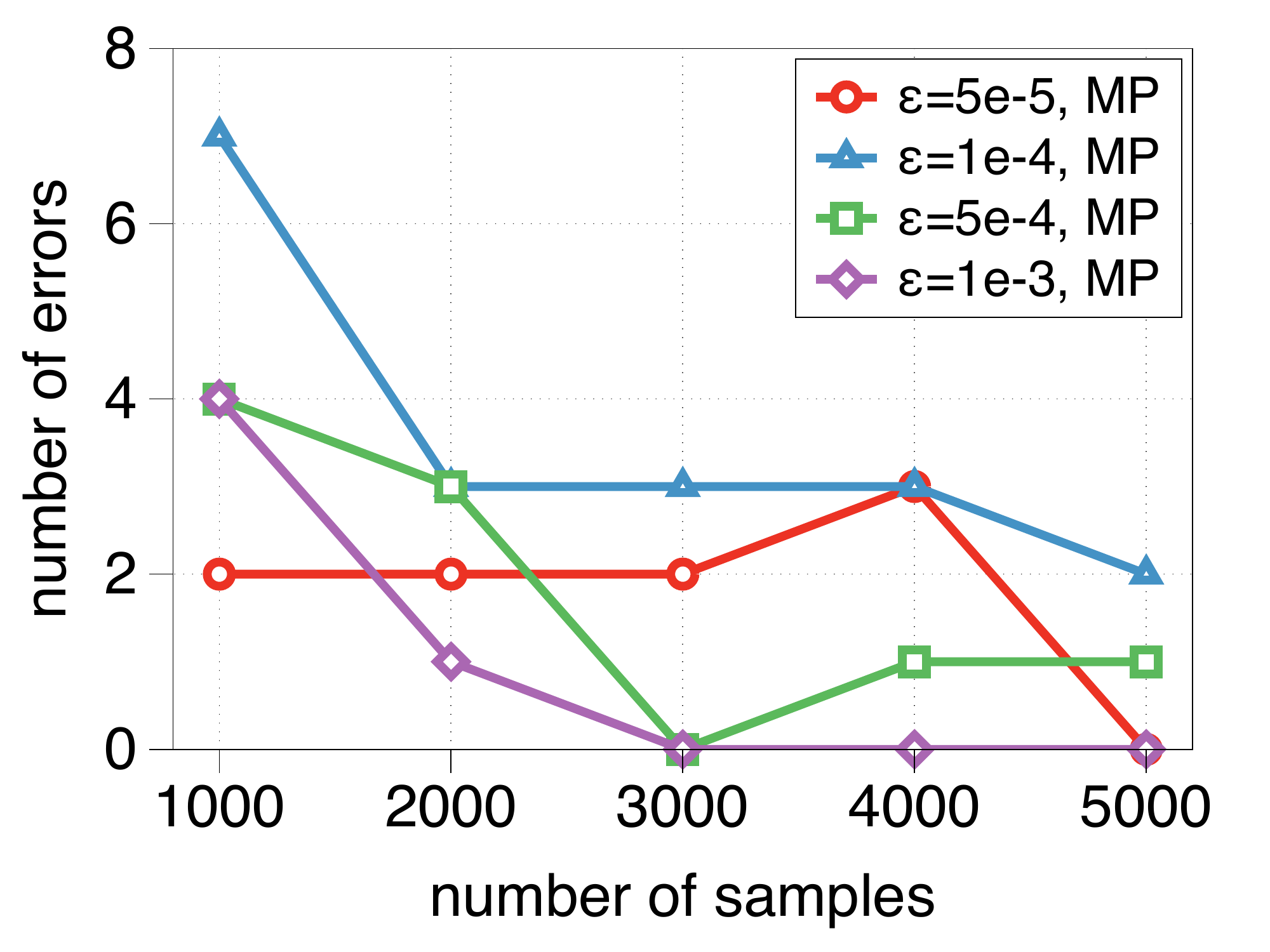}
 \caption{}
 \label{fig:33bweps}
 \end{subfigure}
 ~
 \begin{subfigure}{0.185\textwidth}
 \includegraphics[width=\textwidth]{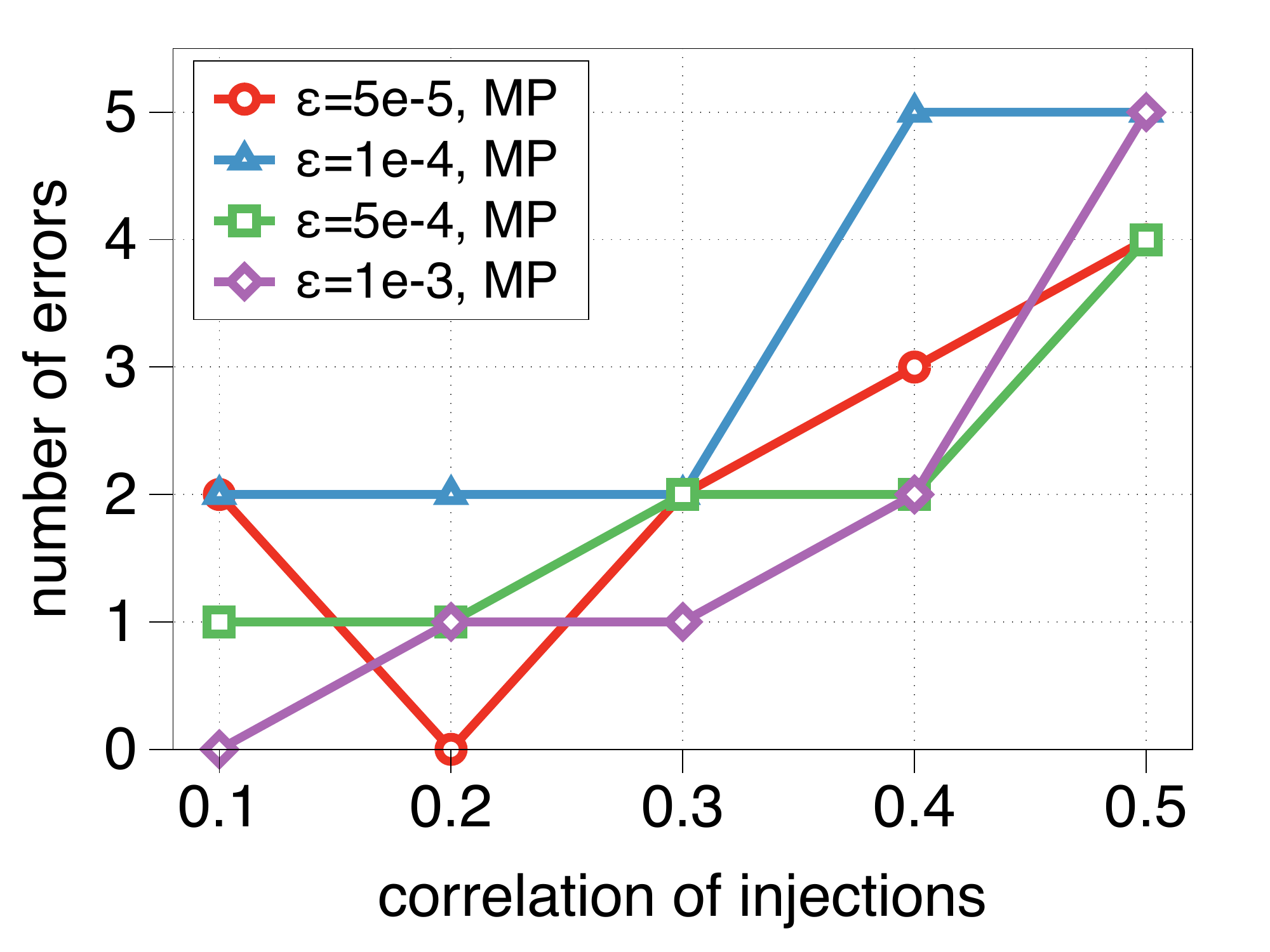}
 \caption{}
 \label{fig:33bwcorr}
 \end{subfigure}
 ~
 \begin{subfigure}{0.185\textwidth}
 \includegraphics[width=\textwidth]{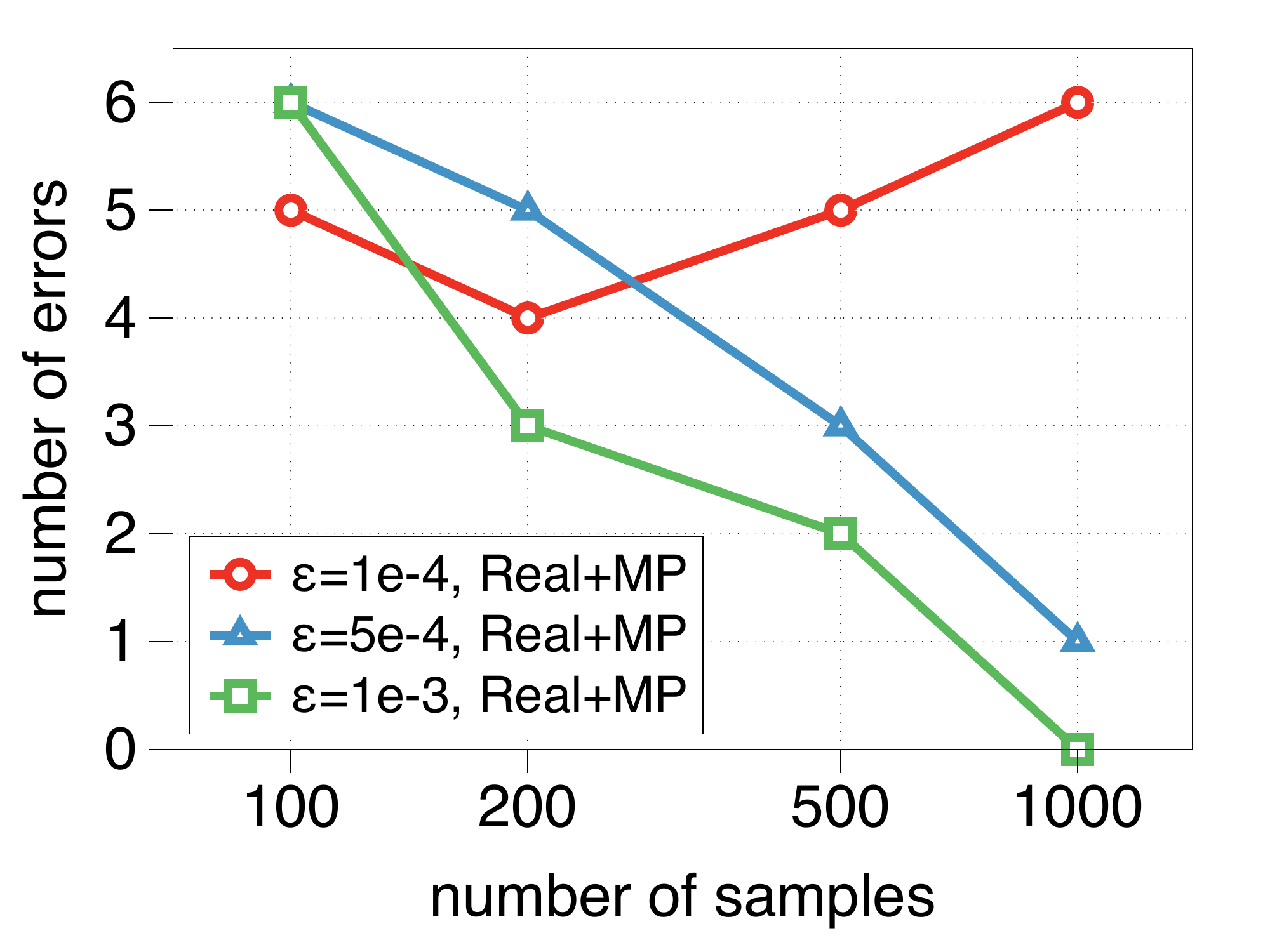}
 \caption{}
 \label{fig:33bwhouse}
 \end{subfigure}
 \caption{Errors in topology estimation of Algorithm \ref{alg:main}
 (a) in $56$ bus model for different power injection standard deviation (b) in $56$ bus model for different values of $\varepsilon$ (c) in $33$ bus model for different values of $\varepsilon$
 (d) in $33$ bus model with different correlation between nodal injections, at 5000 samples (e) in $33$ bus model with real injections. MP, LC and Real implies that the samples are from MATPOWER, LC-PF and real data, respectively.}
 \label{fig:expmain}
\end{figure*}

\noindent{\bf IEEE test cases with non-linear ac samples:}
Here, we discuss more realistic simulations of Algorithm \ref{alg:learningdeep} and \ref{alg:main} on samples generated by the ac power flow model in test cases. We use a modified MATPOWER test case with $33$ nodes \cite{matpower}, and a modified case with $56$ nodes \cite{bolognani2013identification} derived from the IEEE $123$ test feeder \cite{kersting2001radial}. Note that modifications were made to ensure radial topology and that all internal nodes have a minimum degree $\ge3$ (see Assumption \ref{asm:deg3}). The modified grids are illustrated in Fig. \ref{fig:ieee123grid} and Fig. \ref{fig:33bwgrid}. We generate the complex power injections from the independent normal distribution as in the case of the custom models. We obtain the corresponding voltage magnitudes by using ac power flow equations in MATPOWER \cite{zimmerman2011matpower}. We also compare the performance of our algorithm on LC-PF voltages generated with the same complex power injections to see the effect of non-linearity. 

We first show simulation results of Algorithm \ref{alg:learningdeep}, where the input includes complex nodal injection statistics, voltage magnitude at all leaf/end-user nodes, and a set of permissible edges with known impedances from which the operational edges are determined. We consider the $33$ bus model with nodal injection {standard deviation of $10^{-2}$ p.u.} per node. We include 50 additional edges of comparable impedances along with the true operational edges to create the input permissible edge set of 82 edges. Fig. \ref{fig:33bwtau} shows the errors in topology estimation for different values of tolerances with an increasing number of samples. Note that for LC-PF and ac power flow voltage magnitude samples, the errors are comparable. At extremely large sample sizes that represent asymptotic algorithm performance, the errors are close to zero. 

Next, we discuss Algorithm \ref{alg:main} where the input comprises of voltage and injection samples. Under this setting, we measure the performance of our algorithm by varying the number of voltage and injection samples available, the standard deviation of the complex power injection, and the threshold value, $\varepsilon$ used in Algorithm \ref{alg:main}. The experimental results are summarized in Fig. \ref{fig:expmain} where the standard deviation of injections is $10^{-2}$ p.u., $\varepsilon=10^{-4},\tau=1$ for Algorithm \ref{alg:main} unless otherwise noted. To quantify errors in topology estimation, we count the number of edge differences between the recovered topology and the true topology. Fig. \ref{fig:ieee123var} and Fig. \ref{fig:ieee123eps} show our $56$ bus model experimental results. In Fig. \ref{fig:ieee123var}, we observe that the algorithm works similarly for both MATPOWER samples and the LC-PF samples. In Fig. \ref{fig:ieee123eps}, in accordance with observations for custom model experiments, the algorithm performance decreases as the threshold $\varepsilon$ increases. For the $33$ bus model, we perform similar experiments and report results in Fig. \ref{fig:33bweps}. 

Now, we evaluate Algorithm \ref{alg:main} under a more realistic setting where all injections (for both leaf and internal injections) are correlated, i.e., $(p^T,q^T)^T\sim N(0, (1-c)I+c\boldsymbol 1)$ where $c$ denotes the correlation between injections, $I$ denotes the identity matrix, and $\boldsymbol 1$ denotes the matrix consisting of ones. Interestingly, Algorithm \ref{alg:main} even performs well under mild correlation between injections, as presented in Fig. \ref{fig:33bwcorr} for the $33$ bus model. In Fig. \ref{fig:33bwhouse}, we evaluate Algorithm \ref{alg:main} using real load active power data from \cite{disc}, sampled at $15$ minute intervals. We generate reactive power samples from the active loads using a constant power factor and construct the complex-power leaf node injections. The internal node injections are sampled from independent Gaussian distributions as in prior experiments. Given the injections, we generate voltage samples using MATPOWER. Surprisingly, Algorithm \ref{alg:main} outputs accurate estimates at much lower samples sizes, as demonstrated in Fig. \ref{fig:33bwhouse}. In addition, we observe that more errors occur from reconstructing sibling relationships far from the substation node. This may be attributed to the increasing non-linearity as the depth of a grid grows, resulting incorrect impedance distance estimates. This is in line with our observation that the algorithm performs better for the $33$ bus model at low number of samples compared to the $56$ bus model.
\begin{figure}[ht]
 \centering
\includegraphics[width=0.2\textwidth]{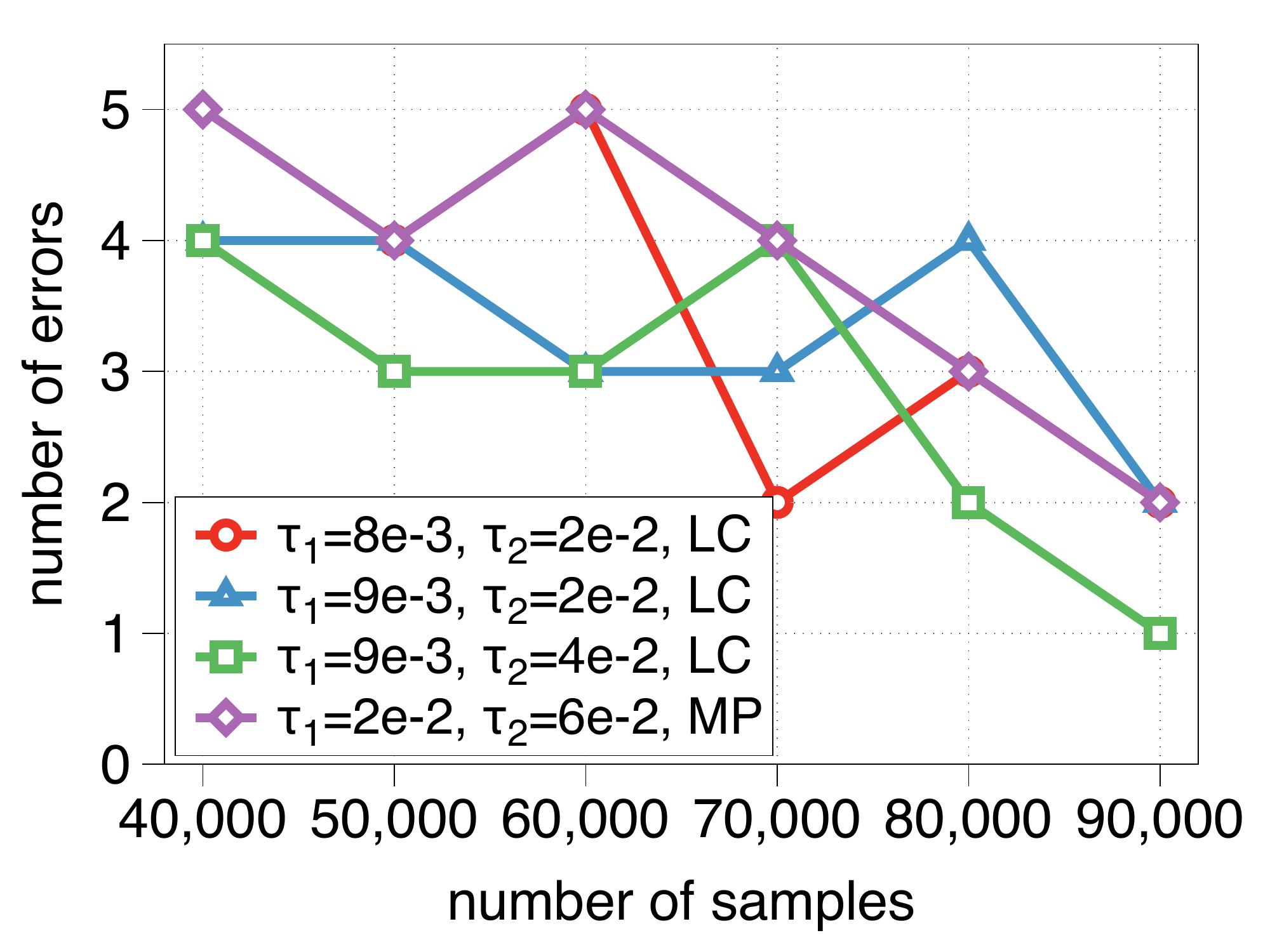}
\caption{Errors in topology estimation of Algorithm \ref{alg:learningdeep} in $33$ bus model for tolerance values $\tau_1$,$\tau_2$. MP, LC implies that the samples are from MATPOWER, LC-PF respectively.}
 \label{fig:33bwtau}
\end{figure}

\noindent\textbf{Sample collection window:} In practice, the time for collecting samples depends on the type of meter and observation window. For example, smart meter data is collected at $5$-$15$ minutes intervals, while (micro-)PMU is collected at a sampling frequency of $30$ Hz. In that regard, the simulated case with correlated injections in the $33$ bus model (Fig.~\ref{fig:33bwcorr}) will take approximately $3$ minutes of PMU observations for 5000 samples. On the other hand, smart meter based estimation (Fig.~\ref{fig:33bwhouse}) will take $8$-$25$ hours  for $100$ real-world samples assuming $5$-$15$ minutes sampling intervals. However, it is worth mentioning that all results of Algorithm \ref{alg:main} assume all node pairs as potential edges. In realistic grids, the set of candidate edges is often limited to a much smaller set. Furthermore, the operator has information about the grid structure in time-intervals preceding the observation window. Both of these may be used to improve the sample performance of learning algorithms, by reducing the line search space and biasing the search towards topologies closer to the prior topology. 

\section{Conclusion and Future Work}\label{sec:conclusions}
In this paper, we present two algorithms that recover topology (and line impedances) using voltage and injection measurements collected only from the end-users/leaf nodes in the radial distribution grid, while all intermediate nodes are unobserved. The first algorithm uses injection statistics at end-users and learns the topology. On the other hand, the second algorithm owing to the presence of injection samples is able to do joint topology and impedance estimation. We show that either algorithm has a computational complexity which scales as $|\mathcal V|^3$. Further, we show that under some mild technical conditions, the second algorithm guarantees to output the correct topology with only $O(|\mathcal V|\log|\mathcal V|)$ samples. We demonstrate the performance of our algorithms through numerical simulations with samples generated from the non-linear ac power flow model in MATPOWER. 

This work opens up several directions for possible extension. We plan to analyze such learning algorithms for three phase power distribution grids under a similar linearized scheme as proposed in \cite{dekathreephase}. While this paper discusses buses with $PQ$ loads with independent nodal injections, non-trivial extensions to systems with correlated injections, voltage regulators and transformers will be analyzed in future work. Finally, we plan to pursue theoretical extensions of this work from radial grids to the case of loopy grids with large girth.

\vspace{-0.075in}
\bibliography{reference,sigproc,FIDVR,SmartGrid,voltage,trees}
\appendix

\subsection{Proof of Lemma \ref{lem:subgaus}}\label{sec:pflem:subgaus}
We start the proof by bounding $\mathbb{E}[e^{\lambda(p_a^2-\mathbb{E}[p_a^2])}]$ utilizing the following lemmas. The proof of Lemma \ref{lem:subgausprod} is presented in Appendix \ref{sec:pflem:subgausprod}.
\begin{lemma}[Lemma 5.5 of \cite{vershynin2010introduction}]\label{lem:subgaus2}
Given a zero mean random variable $X$, the statement $i$ implies the statement $j$ with $K_j=L K_i$ for some constant $L>0$, i.e., they all imply sub-Gaussianity.
\begin{itemize}
 \item[1.]$\mathbb{P}(|X|>t)\le e^{1-t^2/K_1^2}$
 \item[2.]$(\mathbb{E}[|X|^r)^{1/r}\le K_2\sqrt{r}$
 \item[3.]$\mathbb{E}[e^{\lambda X}]\le e^{\lambda^2K_3^2}$
\end{itemize}
\end{lemma}

\begin{lemma}\label{lem:subgausprod}
Let $X,Y$ be independent zero mean random variables satisfying $\mathbb{E}[e^{\lambda X}]\le e^{\lambda^2\sigma_X^2}$,$~\mathbb{E}[e^{\lambda Y}]\le e^{\lambda^2\sigma_Y^2}.$
Then, the following bounds hold
\begin{align*}
&\mathbb{E}[e^{\lambda XY}]\le e^{\lambda^2 (C_1\sigma_X\sigma_Y)^2}\qquad\text{for $|\lambda|\le \frac{C_2}{\sigma_X\sigma_Y}$}\\
&\mathbb{E}[e^{\lambda( X^2-\mathbb{E}[X^2])}]\le e^{\lambda^2 (C_1\sigma_X^2)^2}~~\text{for $|\lambda|\le \frac{C_2}{\sigma_X^2}$}
\end{align*}
for some constants $C_1,C_2$.
\end{lemma}
From the assumption in Theorem \ref{thm:main}, assume that $p_a,q_a$ are sub-Gaussian with sub-Gaussian parameters bounded by a constant $K$.
Under this assumption,
from Lemma \ref{lem:subgaus2}, $\sigma=LK$ satisfies
$$\mathbb{E}[e^{\lambda p_a}]\le e^{\lambda^2\sigma^2}$$
where $L$ is a constant appearing in Lemma \ref{lem:subgaus2}.
From Lemma \ref{lem:subgausprod}, one can observe that there exist absolute constants $C_1,C_2$ such that
$\mathbb{E}[e^{\lambda(p_a^2-\mathbb{E}[p_a^2])}]\le e^{\lambda^2(C_1\sigma^2)^2}$
 for $|\lambda|\le\frac{C_2}{\sigma^2}$.
Note that the same bound holds for $\mathbb{E}[e^{\lambda(q_a^2-\mathbb{E}[q_a^2])}]$.
Now, we address to bound $\mathbb{E}[e^{\lambda(p_aq_a-\mathbb{E}[p_aq_a])}]$.
To this end, we consider the following decomposition
\begin{equation}\label{eq:paqa}
p_aq_a=\frac12(p_a+q_a)^2-\frac12p_a^2-\frac12q_a^2
\end{equation}
and introduce the below lemma.
\begin{lemma}\label{lem:subexpsum}
Let $X,Y$ be zero mean random variables satisfying \begin{align*}
&\mathbb{E}[e^{\lambda X}]\le e^{\lambda^2\sigma_X^2}\quad\text{for $|\lambda|\le B_X$}\\
&\mathbb{E}[e^{\lambda Y}]\le e^{\lambda^2\sigma_Y^2}\quad\text{for $|\lambda|\le B_Y$}.
\end{align*}
Then, $\mathbb{E}[e^{\lambda (X+Y)}]\le e^{\lambda^2(2\sigma_X^2+2\sigma_Y^2)}$
for $|\lambda|\le\frac12\min(B_X,B_Y)$.
In addition, if $X,Y$ are sub-Gaussian, i.e., $B_X=B_Y=\infty$, then $X+Y$ is also sub-Gaussian with
\begin{align*}
&\mathbb{E}[e^{\lambda (X+Y)}]\le e^{\lambda^2(\sigma_X^2+\sigma_Y^2)}~~\text{if $X,Y$ are independent}\\
&\mathbb{E}[e^{\lambda (X+Y)}]\le e^{\lambda^2(\sigma_X+\sigma_Y)^2}~\text{otherwise}.
\end{align*}
\end{lemma}
The proof of Lemma \ref{lem:subexpsum} is given in Appendix \ref{sec:pflem:subexpsum}. Lemma \ref{lem:subexpsum} directly implies that $p_a+q_a$ is sub-Gaussian satisfying
$$\mathbb{E}[e^{\lambda(p_a+q_b)}]\le e^{4\lambda^2\sigma^2}.$$
Further, from Lemma \ref{lem:subgausprod}, the following bound holds
$$\mathbb{E}[e^{\lambda((p_a+q_a)^2-\mathbb{E}[(p_a+q_a)^2])}]\le e^{\lambda^2(4C_1\sigma^2)^2}$$
for $|\lambda|\le\frac{C_2}{4\sigma^2}$.
Using the above bound, Lemma \ref{lem:subexpsum}, and Eq. \eqref{eq:paqa}, the bound for $\mathbb{E}[e^{\lambda(p_aq_a-\mathbb{E}[p_aq_a])}]$ can be derived as
\begin{align}
&\mathbb{E}[e^{\lambda(p_aq_a-\mathbb{E}[p_aq_a])}]\notag\\
&=\mathbb{E}[e^{\frac\lambda2((p_a+q_a)^2-\mathbb{E}[(p_a+q_a)^2]-p_a^2+\mathbb{E}[p_a^2]-q_a^2+\mathbb{E}[q_a^2])}]\notag\\
&\le e^{\lambda^2(C_1^\prime\sigma^2)^2}\label{eq:paqabd}
\end{align}
for $|\lambda|\le\frac{C_2^\prime}{\sigma^2}$ for some constants $C_1^\prime, C_2^\prime$. 

So far, we found bounds for $p_a^2,q_a^2$ and $p_aq_a$.
Now, we begin to bound $\mathbb{E}[e^{\lambda(v_ap_b-\mathbb{E}[v_ap_b])}]$.
As in obtaining the bound for $p_aq_b$, we first decompose $v_ap_b$ using Eq. \eqref{eq:lcpf2} as
\begin{align*}
v_ap_b&=\sum_{c\in\mathcal V}H_{1/r}^{-1}(a,c)p_bp_c+H_{1/x}^{-1}(a,c)p_bq_c\\
&=H_{1/r}^{-1}(a,b)p_b^2+p_b\bar{p}_b+H_{1/x}^{-1}(a,b)p_bq_b+p_b\bar{q}_b, \text{~where}
\end{align*}
$\bar{p}_b=\sum_{c\in\mathcal V\setminus\{b\}}H_{1/r}^{-1}(a,c)p_c,\quad \bar{q}_b=\sum_{c\in\mathcal V\setminus\{b\}}H_{1/x}^{-1}(a,c)q_c.$
As done before, we will bound $\mathbb{E}[e^{\lambda p_b\bar{p}_b}]$ and $\mathbb{E}[e^{\lambda p_b\bar{q}_b}]$.
Let us define $H_{\max}=\max_{a,b\in\mathcal V}(\max(|H_{1/r}^{-1}(a,b)|,|H_{1/x}^{-1}(a,b)|))$.
Since we assume the constant depth of the power grid and bounded line parameters in Theorem \ref{thm:main}, $H_{\max}$ is constantly bounded due to Lemma \ref{lemma1}.
Using this and Lemma \ref{lem:subgaus2}, we bound 
$$\mathbb{E}[e^{\lambda H_{1/r}^{-1}(a,c)p_c}]\le e^{(\lambda H_{1/r}^{-1}(a,c))^2\sigma^2}\le e^{\lambda^2 H_{\max}^2\sigma^2}.$$
Moreover, using Lemma \ref{lem:subexpsum}, we also bound
$$\mathbb{E}[e^{\lambda \bar{p_b}}]\le e^{\lambda^2(|\mathcal{V}|-1)H_{\max}^2\sigma^2}\le e^{\lambda^2|\mathcal{V}|H_{\max}^2\sigma^2}.$$
Since $p_b,\bar{p}_b$ are independent, using Lemma \ref{lem:subgausprod}, one can derive $\mathbb{E}[e^{\lambda p_b\bar{p_b}}]\le e^{\lambda^2(C_1\sqrt{|\mathcal{V}|}H_{\max}\sigma^2)^2}$
for $|\lambda|\le\frac{C_2}{\sqrt{|\mathcal V|}H_{\max}\sigma^2}$.
Finally, using Lemma \ref{lem:subexpsum}, the following bound holds
\begin{align}
&\mathbb{E}[e^{\lambda(v_ap_b-\mathbb{E}[v_ap_b])}]\notag\\
&=\mathbb{E}[e^{\lambda(H_{1/r}^{-1}(a,b)(p_b^2-\mathbb{E}[p_b^2])+H_{1/x}^{-1}(a,b)(p_bq_b-\mathbb{E}[p_bq_b])+p_b\bar{p}_b+p_b\bar{q}_b)}]\notag\\
&\le e^{\lambda^2(C_1^{\prime\prime}\sqrt{|\mathcal V|}\sigma^2)^2}\quad\text{for\quad  $|\lambda|\le{C_2^{\prime\prime}}/{\sqrt{|V|}\sigma^2}$}\label{eq:vapbbd}
\end{align}
for some constants $C_1^{\prime\prime},C_2^{\prime\prime}$.
Note that same bound holds for $\mathbb{E}[e^{\lambda(v_aq_b-\mathbb{E}[v_aq_b])}]$. Choosing $\alpha=\max(C_1^\prime,C_1^{\prime\prime})\sigma^2, M=\frac{\min(C_2^\prime,C_2^{\prime\prime})}{\sigma^2}$ completes the proof of Lemma \ref{lem:subgaus}.
\subsection{Proof of Lemma \ref{lem:bernstein}}\label{sec:pflem:bernstein}
Before starting the proof, we note that the proof is analogous to the proof of Proposition 5.16 in \cite{vershynin2010introduction}.
Let $S=\sum_iX_i$. In this proof,
We split the cases for $S\ge t$ and $-S\ge t$.
To this end, we bound the probability of $S\ge t$.
\begin{align}
\mathbb{P}(S\ge t)&=\mathbb{P}(e^{\lambda S}\ge e^{\lambda t})\le e^{-\lambda t}\mathbb{E}[e^{\lambda S}]\label{eq:berstein2}\\
&=e^{-\lambda t}\prod_{i=1}^n\mathbb{E}[e^{\lambda X_i}]\\&\le e^{-\lambda t}\prod_{i=1}^n e^{\lambda^2\sigma^2}~\text{for $|\lambda|\le B$}\label{eq:berstein4}\\
&=e^{-\lambda t+\lambda^2\sigma^2n}~\text{for $|\lambda|\le B$}\label{eq:berstein5}\\
&\le\exp\left(-\min\left(\frac{Bt}{2},\frac{t^2}{4\sigma^2n}\right)\right)\label{eq:berstein6}
\end{align}
Eq. \eqref{eq:berstein2} is from Markov's inequality. Eq. \eqref{eq:berstein4} is from the assumption of Lemma \ref{lem:bernstein}.
Eq. \eqref{eq:berstein6} is from choosing $\lambda=\min(B,t/(2\sigma^2n))$. One can obtain the same bound for $\mathbb{P}(-S\ge t)$. Applying union bound on $\mathbb{P}(-S\ge t)$, $\mathbb{P}(S\ge t)$ leads us to the result of Lemma \ref{lem:bernstein}.
This completes the proof of Lemma \ref{lem:bernstein}.

\subsection{Proof of Lemma \ref{lem:subgausprod}}\label{sec:pflem:subgausprod}
We first derive the bound for $\mathbb{E}[e^{\lambda XY}]$.
\begin{align}
&\mathbb{E}[e^{\lambda XY}]=1+\lambda \mathbb{E}[X]\mathbb{E}[Y]+\sum_{r=2}^\infty\frac{\lambda^r\mathbb{E}[X^r]\mathbb{E}[Y^r]}{r!}\label{eq:subgausprod1}\\
&\le1+\sum_{r=2}^\infty\frac{\lambda^r\mathbb{E}[|X|^r]\mathbb{E}[|Y|^r]}{r!}\label{eq:subgausprod3}\\
&\le1+\sum_{r=2}^\infty\frac{\lambda^r(L\sigma_X\sqrt{r})^r(L\sigma_Y\sqrt{r})^r}{r!}\label{eq:subgausprod4}\\
&\le1+\sum_{r=2}^\infty\frac{(\lambda L^2\sigma_X\sigma_yr)^r}{\sqrt{2\pi r}r^re^{-r}}=1+\sum_{r=2}^\infty\frac{(\lambda eL^2\sigma_X\sigma_y)^r}{\sqrt{2\pi r}}\label{eq:subgausprod5}\\
&\le1+\frac1{\sqrt{2\pi}}\sum_{r=2}^\infty{(\lambda eL^2\sigma_X\sigma_y)^r}\label{eq:subgausprod7}\\
&=1+\frac{(\lambda eL^2\sigma_X\sigma_y)^2}{\sqrt{2\pi}(1-\lambda eL^2\sigma_X\sigma_y)}~\text{for $|\lambda|<\frac1{eL^2\sigma_X\sigma_Y}$}\label{eq:subgausprod8}\\
&\le1+{(2\lambda eL^2\sigma_X\sigma_y)^2}~\text{for $|\lambda|\le\left(1-\frac1{4\sqrt{2\pi}}\right)\frac{1}{eL^2\sigma_X\sigma_Y}$}\label{eq:subgausprod9}\\
&\le\exp(\lambda^2(2eL^2\sigma_X\sigma_y)^2)\label{eq:subgausprod10}
\end{align}
Eq. \eqref{eq:subgausprod1} is from the Taylor series expansion and the independence of $X,Y$.
Eq. \eqref{eq:subgausprod4} is from Lemma \ref{lem:subgaus2} and $L$ is an absolute constant appearing in Lemma \ref{lem:subgaus2}.
Eq. \eqref{eq:subgausprod5} is from the lowerbound of Stirling's approximation $\sqrt{2\pi r}r^re^{-r}\le r!$.
Eq. \eqref{eq:subgausprod7} is obtained by deleting $\sqrt{r}$ in the denominator.
Eq. \eqref{eq:subgausprod8} is from the sum of power series.
Eq. \eqref{eq:subgausprod9} is from $\frac1{\sqrt{2\pi}(1-\lambda eL^2\sigma_X\sigma_Y)}\le 4$ when $|\lambda|\le\left(1-\frac1{4\sqrt{2\pi}}\right)\frac1{eL^2\sigma_X\sigma_Y}$.
Eq. \eqref{eq:subgausprod10} is from $1+x\le e^x$.

Now, we consider the bound for $\mathbb{E}[e^{\lambda X^2}]$. For this bound, We refer Appendix B of \cite{honorio2014tight} which states that
$$\mathbb{E}[e^{\lambda (X^2-\mathbb{E}[X^2])}]\le e^{\lambda^2(8\sigma_X)^2}~\text{for $|\lambda|\le\frac1{8\sigma_X^2}$}.$$
Choosing $$C_1=\max(2eL^2,8),~C_2=\min\left(\left(1-\frac1{4\sqrt{2\pi}}\right)\frac1{eL^2},8\right)$$ completes the proof of Lemma \ref{lem:subgausprod}.

\subsection{Proof of Lemma \ref{lem:subexpsum}}\label{sec:pflem:subexpsum}
First, we consider the case that at least one of $B_X,B_Y$ is bounded.
To this end, we derive the following bound which directly leads us to the first result of Lemma \ref{lem:subexpsum}.
\begin{align}
&\mathbb{E}[e^{\lambda(X+Y)}]\le\left(\mathbb{E}[(e^{\lambda X})^2]\right)^{1/2}\left(\mathbb{E}[(e^{\lambda Y})^2]\right)^{1/2}\label{eq:subexpsum1}\\
&\le (e^{(2\lambda)^2\sigma_X^2})^{1/2}(e^{(2\lambda)^2\sigma_Y^2})^{1/2}~\text{for $|\lambda|\le\frac12\min(B_X,B_Y$)}\label{eq:subexpsum2}\\
&=e^{\lambda^2(2\sigma_X^2+2\sigma_Y^2)}~\text{for $|\lambda|\le\frac12\min(B_X,B_Y$)}\label{eq:subexpsum3}
\end{align}
Here, Eq. \eqref{eq:subexpsum1} is from H\"{o}lder's inequality. Eq. \eqref{eq:subexpsum2} is from the assumption of Lemma \ref{lem:subexpsum}. 

Now, we consider the case when $B_X=B_Y=\infty$. When $X,Y$ are independent, the result is trivial.
When $X,Y$ are dependent, the proof is analogous to the proof of Theorem 2.7 of \cite{rivasplata2012subgaussian}, therefore we omit the proof.
This completes the proof of Lemma \ref{lem:subexpsum}.

 \vspace{-28pt}
 \begin{IEEEbiography}[{\includegraphics[width=.8in,height=1in,clip,keepaspectratio]{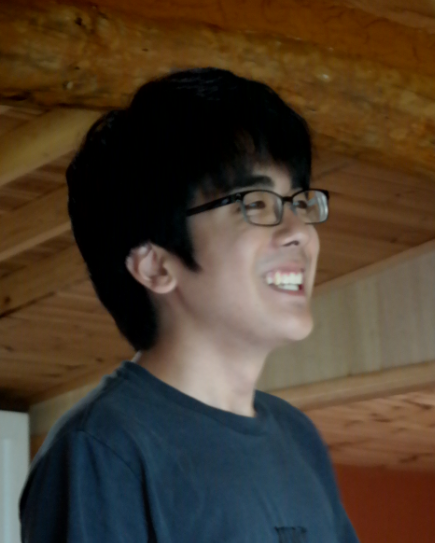}}]{Sejun park}
       is a Ph.D. student  in the School of Electrical  Engineering at Korea  Advanced Institute of  Science and Technology  (KAIST), advised by Prof.  Jinwoo Shin. He has been worked on developing provable inference and learning algorithms for probabilistic graphical models and power distribution grids. His current research interests are discrete problems associated with neural networks including neural network pruning and discrete neural networks.
    \end{IEEEbiography}
    \vspace{-34pt}
     \begin{IEEEbiography}[{\includegraphics[width=.8in,height=1in,clip,keepaspectratio]{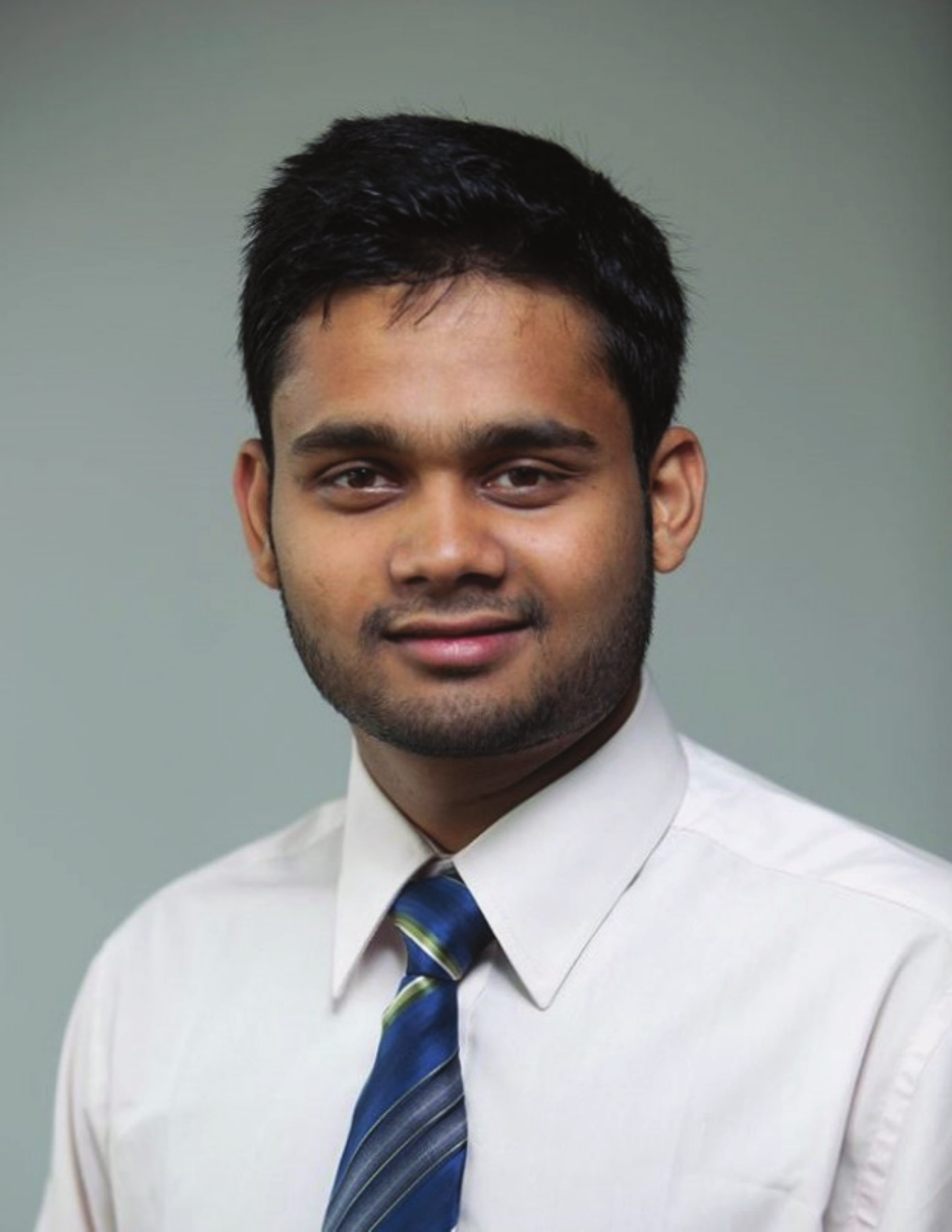}}]{Deepjyoti Deka}
       is a staff scientist in the Applied Mathematics and Plasma Physics group of the Theoretical Division at Los Alamos National Laboratory (LANL), where he was previously a postdoctoral research associate at the Center for Nonlinear Studies (CNLS). His research interests include data-analysis of power grid structure, operations and security, and optimization in social and physical networks. At LANL, Dr. Deka serves as a co-principal investigator for DOE projects on machine learning in distribution systems and in cyber-physical security. Before joining the laboratory he received the M.S. and Ph.D. degrees in electrical engineering from the University of Texas, Austin, TX, USA, in 2011 and 2015, respectively. He completed his undergraduate degree in electrical engineering from IIT Guwahati, India  in 2009 with an institute silver medal as the best outgoing student of the department.
    \end{IEEEbiography} \vspace{-34pt}
        \begin{IEEEbiography}[{\includegraphics[width=.8in,height=1in,clip,keepaspectratio]{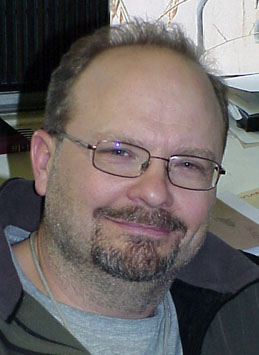}}]{Scott Backhaus}
        is currently the Cryogenics Team Lead in the Quantum Electromagnetics Group at the National Institute of Standards and Technology. He is also an independent consultant to Camus Energy, a startup company focused on monitoring and control of electrical distribution networks. He recently served as the Coordinator for Electromagnetic Pulse (EMP) and Geomagnetic Disturbance (GMD) impacts on critical infrastructure systems for the Department of Homeland Security’s (DHS) Cybersecurity and Infrastructure Security Agency. He previously served in multiple roles in his 20-plus years at Los Alamos National Laboratory (LANL), including Program Manager for Office of Electricity, Program Manager for DHS Critical Infrastructure, principal investigator for several LANL projects funded by the Office of Electricity, and team lead for LANL’s component of the DHS National Infrastructure Simulation and Analysis Group.  He received his Ph.D. in Physics in 1997 from the University of California at Berkeley in the area of macroscopic quantum behavior of superfluid He(3) and He(4).
    \end{IEEEbiography}
     \vspace{-34pt}
        \begin{IEEEbiography}[{\includegraphics[width=.8in,height=1in,clip,keepaspectratio]{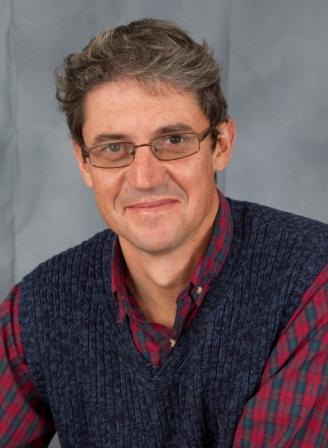}}]{Michael Chertkov's} area of focus is mathematics, including statistics and data science, applied to physical, engineered and other systems. Dr. Chertkov received his Ph.D.in physics from the Weizmann Institute of Science in 1996, and his M.Sc. in physics from Novosibirsk State University in 1990. After his Ph.D., Dr. Chertkov spent three years at Princeton University as a R. H. Dicke Fellow in the Department of Physics. He joined Los Alamos National Lab in 1999, initially as a J.R.Oppenheimer Fellow in the Theoretical Division, and continued as a Technical Staff Member. In 2019, Dr. Chertkov joined the University of Arizona as a Professor of Mathematics and leads the Interdisciplinary Graduate Program in Applied Mathematics. He is a fellow of the American Physical Society (APS) and a senior member of IEEE. 
    \end{IEEEbiography}
   \end{document}